\documentclass[11pt]{article}
\usepackage{cite}
\usepackage{amsmath,amssymb,amsfonts}
\usepackage{algorithmic}
\usepackage{graphicx}
\usepackage{textcomp}
\usepackage{xcolor}
\def\BibTeX{{\rm B\kern-.05em{\sc i\kern-.025em b}\kern-.08em
    T\kern-.1667em\lower.7ex\hbox{E}\kern-.125emX}}

\usepackage{color}
\usepackage{array}
\usepackage{comment}
\usepackage{multirow}
\usepackage{algorithm}
\usepackage{amsthm}
\usepackage{subcaption}
\usepackage{here}
\usepackage{authblk}

\newtheorem{definition}{Definition}

\newtheorem{theorem}{Theorem}

\newcommand{\RIGHT}{\rightarrow}
\newcommand{\LEFT}{\leftarrow}
\newcommand{\UP}{\uparrow}
\newcommand{\DOWN}{\downarrow}

\newcommand{\ALG}{{\cal A}}
\newcommand{\PROB}{{\cal P}}

\newcommand{\W}{{\textsf W}}
\newcommand{\G}{{\textsf G}}
\newcommand{\B}{{\textsf B}}
\newcommand{\V}{{\cal V}}

\title{Terminating Grid Exploration with Myopic Luminous Robots}
\author[1]{Shota Nagahama}
\author[1]{Fukuhito Ooshita}
\author[1]{Michiko Inoue}
\affil[1]{Nara Institute of Science and Technology, Japan}

\date{}

\begin{document}



\maketitle

\begin{abstract}
We investigate the terminating grid exploration for autonomous myopic luminous robots. Myopic robots mean that they can observe nodes only within a certain fixed distance, and luminous robots mean that they have light devices that can emit colors. First, we prove that, in the semi-synchronous and asynchronous models, three myopic robots are necessary to achieve the terminating grid exploration if the visible distance is one. Next, we give fourteen algorithms for the terminating grid exploration in various assumptions of synchrony (fully-synchronous, semi-synchronous, and asynchronous models), visible distance, the number of colors, and a chirality. Six of them are optimal in terms of the number of robots.
\end{abstract}


\section{Introduction}

\subsection{Background and motivation}
Many studies about cooperation of autonomous mobile robots have been conducted in the field of distributed computing.
These studies focus on the minimum capabilities of robots that permit to achieve a given task.
To model operations of robots, the \emph{Look-Compute-Move (LCM) model}~\cite{Suzuki99:Distributed} is commonly used.
In the LCM model, each robot repeats cycles of Look, Compute, and Move phases. In the Look phase, the robot observes positions of other robots. In the Compute phase, the robot executes its algorithm using the observation as its input, and decides whether it moves somewhere or stays idle. In the Move phase, it moves to a new position if the robot decided to move in the Compute phase.
To consider minimum capabilities, most studies assume that robots are \emph{identical} (\emph{i.e.}, robots execute the same algorithm and have no identifier), \emph{oblivious} (\emph{i.e.}, robots have no memory to record their past history), and \emph{silent} (\emph{i.e.}, robots do not have communication capabilities). 
Furthermore, they have \emph{no global compass}, \emph{i.e.}, they do not agree on the directions.
Based on the LCM model, previous works clarified solvability of many tasks such as exploration, gathering, and pattern formation in continuous environments (aka two- or three-dimensional Euclidean space) and discrete environments (aka graph networks) (see a survey~\cite{flocchini2019distributed}).

In this paper, we focus on \emph{exploration} in graph networks, which is one of the most central tasks for mobile robots.
Two variants of exploration tasks have been well studied: the \emph{perpetual} exploration requires every robot to visit every node infinitely many times, and the \emph{terminating} exploration requires robots to terminate after every node is visited by a robot at least once. 
During the last decade, many works have considered the perpetual and terminating exploration on the assumption that each robot has unlimited visibility, \emph{i.e.}, it observes all other robots in the network.
The perpetual exploration has been studied for rings~\cite{Blin10:Exclusive} and grids~\cite{Bonnet11:Asynchronous}.
The terminating exploration has been studied for lines~\cite{flocchini2011how}, rings~\cite{Devismes13:Optimal,Flocchini13:Computing, lamani2010optimal}, trees~\cite{Flocchini10:Remembering}, finite grids~\cite{Devismes12:Optimal,Devismes21:Terminating}, tori~\cite{Devismes19:Optimal}, and arbitrary networks~\cite{Chalopin10:Network}. 
However, the capability of the unlimited visibility seems powerful and somewhat contradicts the principle of weak mobile robots. 
For this reason, some studies consider the more realistic case of \emph{myopic} robots~\cite{Datta13:Ring,Datta13:Ring23}. 
A myopic robot has limited visibility, \emph{i.e.}, it can see nodes (and robots on them) only within a certain fixed distance $\phi$. 
Datta et al.\ studied the terminating exploration of rings for $\phi=1$~\cite{Datta13:Ring} and $\phi=2,3$~\cite{Datta13:Ring23}. 
Not surprisingly, since myopic robots are weaker than non-myopic robots, many impossibility results are given for myopic robots.

To improve the task solvability, myopic robots with persistent visible light~\cite{Das16:Autonomous}, called myopic \emph{luminous} robots, have attracted a lot of attention.
Each myopic luminous robot is equipped with a light device that can emit a constant number of colors to other robots, a single color at a time. 
The light color is persistent, \emph{i.e.}, it is not automatically reset at the end of each cycle, and hence it can be used as a constant-space memory.

Ooshita and Tixeuil~\cite{Ooshita21:Ring} studied the perpetual and terminating exploration of rings for $\phi=1$ in the synchronous (FSYNC), semi-synchronous (SSYNC), and asynchronous (ASYNC) models.
They showed that the number of robots required to achieve the tasks can be reduced compared to non-luminous robots.
Nagahama et al.~\cite{nagahama2019ring} 
studied the same problem in case of $\phi\geq 2$ and showed that, in the SSYNC and ASYNC models, the number of robots required to achieve the tasks can be reduced compared to the case of $\phi=1$.

Bramas et al.\ studied the exploration of an infinite grid with myopic luminous and non-luminous robots in the FSYNC model~\cite{Bramas20:Infinite,Bramas20:Finding}.
Here they propose algorithms so that every node of an infinite grid is visited by a robot at least once.
In \cite{Bramas20:Infinite} robots agree on a \emph{common chirality}, \emph{i.e.},
robots agree on common clockwise and counterclockwise directions.
Bramas et al.~\cite{bramas2021optimal} also studied the perpetual exploration of a (finite) grid with myopic luminous and non-luminous robots in the FSYNC model on the assumption that robots agree on a common chirality.
Algorithms proposed in \cite{bramas2021optimal} have additional nice properties: they work even if robots are opaque (\emph{i.e.}, a robot is able to see another robot only if no other robot lies in the line segment joining them), and they are exclusive (\emph{i.e.}, no two robots occupy a single node during the execution).
This work also describes the way to extend their algorithms to acheive the terminating exploration and/or to work in the SSYNC and ASYNC models.
More concretely, this gives three algorithms to achieve the terminating exploration of a grid in case of a common chirality: algorithms for two robots with $\phi=1$ and six colors in the FSYNC model, two robots with $\phi=2$ and five colors in the FSYNC model, and two robots with $\phi=2$ and six colors in the SSYNC and ASYNC models.
However algorithms with a fewer number of colors or no common chirality are not known yet.

\subsection{Our contributions}
We focus on the terminating exploration of a (finite) grid with myopic luminous and non-luminous robots, and clarify lower and upper bounds of the required number of robots in various assumptions of synchrony, visible distance $\phi$, the number of colors, and a chirality.
Table \ref{table:summary} summarizes our contributions.

First, we prove that, in the SSYNC and ASYNC models, three myopic robots are necessary to achieve the terminating exploration of a grid if $\phi=1$ holds.
Note that this lower bound also holds for the perpetual exploration because we prove that robots cannot visit some nodes of a grid in this case.
Other lower bounds in Table \ref{table:summary} are given by Bramas et al.~\cite{bramas2021optimal}.
They are originally given as impossibility results for the perpetual exploration, however they still hold for the terminating exploration.
This is because Bramas et al.\ prove that, if the number of robots is smaller in each assumption, robots cannot visit some nodes.

Second, we propose algorithms to achieve the terminating exploration of a grid in various assumptions in Table \ref{table:summary}.
To the best of our knowledge, they are the first algorithms that achieve the terminating exploration of a grid by myopic robots with at most three colors and/or with no common chirality.
In addition, six proposed algorithms are optimal in terms of the number of robots.

\begin{table}[tb]
\centering
\caption{Terminating grid exploration with myopic robots. Notation $\phi$ represents the visible distance of a robot, $\ell$ represents the number of colors, and $*$ means the number of robots is minimum.}
\label{table:summary}
\begin{center}
\begin{tabular}{|c|c|c|c|>{\centering}p{1cm}c|>{\centering}p{1cm}c|}
\hline
\multirow{2}{*}{Synchrony} &
\multirow{2}{*}{$\phi$} &
\multirow{2}{*}{$\ell$} &
Common & 
\multicolumn{4}{c|}{\#required robots} \\ \cline{5-8}
 & & & chirality & 
\multicolumn{2}{c|}{Lower bound} & \multicolumn{2}{c|}{Upper bound} \\ \hline
\multirow{8}{*}{FSYNC} & \multirow{4}{*}{2} & \multirow{2}{*}{2} & yes & 2 & \cite{bramas2021optimal}& $\mathbf{2^*}$ & \S\,\ref{secF22T2}\\ \cline{4-8} 
 & & & no & 2 & \cite{bramas2021optimal}& $\mathbf{3}$ &\S\,\ref{secF22F3}\\ \cline{3-8} 
 & & \multirow{2}{*}{1} & yes & 3 & \cite{bramas2021optimal} & $\mathbf{3^*}$ & \S\,\ref{secF21T3}\\ \cline{4-8} 
 & & & no & 3 & \cite{bramas2021optimal} & $\mathbf{4}$  & \S\,\ref{secF21F4}\\ \cline{2-8} 
 & \multirow{4}{*}{1} & \multirow{2}{*}{3} & yes & 2 & \cite{bramas2021optimal}& $\mathbf{2^*}$ & \S\,\ref{secF13T2}\\ \cline{4-8} 
 & & & no & 2 & \cite{bramas2021optimal}& $\mathbf{4}$ & \S\,\ref{secF13F4}
 \\ \cline{3-8} 
 & & \multirow{2}{*}{2} & yes & 3 & \cite{bramas2021optimal} & $\mathbf{3^*}$ & \S\,\ref{secF12T3}\\ \cline{4-8} 
 & & & no & 3 & \cite{bramas2021optimal} & $\mathbf{5}$ & \S\,\ref{secF12F5}
 \\ \hline
 & \multirow{4}{*}{2} & \multirow{2}{*}{3} & yes & 2 & \cite{bramas2021optimal}& $\mathbf{2^*}$ &\S\,\ref{secA23T2}\\ \cline{4-8} 
 & & & no & 2 & \cite{bramas2021optimal}& $\mathbf{3}$ & \S\,\ref{secA23F3}\\ \cline{3-8} 
 SSYNC & & \multirow{2}{*}{2} & yes & 2 & \cite{bramas2021optimal}& $\mathbf{3}$ & \S\,\ref{secA22T3}
 \\ \cline{4-8} 
 ASYNC & & & no & 2 & \cite{bramas2021optimal}& $\mathbf{4}$ & \S\,\ref{secA22F4}
 \\ \cline{2-8} 
 & \multirow{2}{*}{1} & \multirow{2}{*}{3} & yes & $\mathbf{3}$ & \S\,\ref{secImpSSYNC} & $\mathbf{3^*}$ & \S\,\ref{secA13T3}\\ \cline{4-8} 
 & & & no & $\mathbf{3}$ & \S\,\ref{secImpSSYNC} & $\mathbf{6}$ & \S\,\ref{secA13F6}
 \\ \hline
\end{tabular}
\end{center}
\end{table}

\section{Preliminaries}

\subsection{System model}
The system consists of $k$ mobile robots and a simple connected graph $G=(V, E)$, where $V$ is a set of nodes and $E$ is a set of edges.
In this paper, we assume that $G$ is a finite $m\times n$ grid (or a grid, for short) where $m$ and $n$ are two positive integers, \textit{i.e.}, $G$ satisfies the following conditions:
\begin{itemize}
  \item $V = \{v_{i,j} \,|\, i\in\{0,1,\ldots,m-1\},\, j\in\{0,1,\ldots,n-1\}\}$
  \item $E = \{(v_{i, j}, v_{i', j'}) \,|\, v_{i, j}, v_{i', j'}\in V,\, |i-i'|+|j-j'|=1\}$
\end{itemize}
The indices of nodes are used for notation purposes only and robots do not know them.
Neither nodes nor edges have identifiers or labels, and consequently robots cannot distinguish nodes and cannot distinguish edges. 
Robots do not know $m$ or $n$. 
Figure \ref{global-direction} shows global directions labeled by North, East, South, and West on a grid.
Note that these directions are used only for explanations, and robots cannot access the global directions.
Each robot is on a node of $G$ at each instant.
When a robot $r$ is on a node $v$, we say $r$ \textit{occupies} $v$ and $v$ \textit{hosts} $r$.
The distance between two nodes is the number of edges in a shortest path between the nodes. The distance between two robots $r_1$ and $r_2$ is the distance between two nodes occupied by $r_1$ and $r_2$. Two robots $r_1$ and $r_2$ are neighbors if the distance between $r_1$ and $r_2$ is one. 

\begin{figure}[tb]
\begin{center}
\includegraphics[scale=1.0]{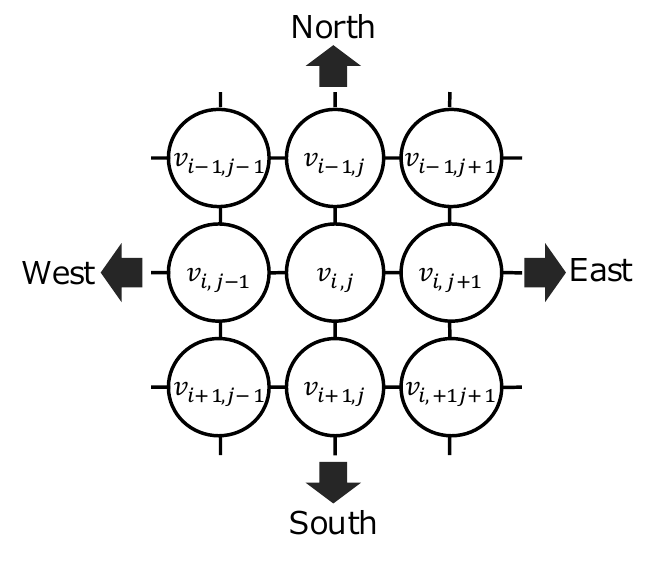}
\end{center}
\caption{Global directions on a grid}
\label{global-direction}
\end{figure}

Robots we consider have the following characteristics and capabilities. Robots are \textit{identical}, that is, robots execute the same deterministic algorithm and do \textit{not} have unique identifiers. Robots are \textit{luminous}, that is, each robot has a light (or state) that is visible to itself and other robots. A robot can choose the color of its light from a discrete set $Col$. When the set $Col$ is finite, $\ell$ denotes the number of available colors (\textit{i.e.}, $\ell=|Col|$). 
Robots have no other persistent memory and cannot remember the history of past actions. 
Each robot can communicate by observing positions and colors of other robots (for collecting information), and by changing its color and moving (for sending information).
Robots are \textit{myopic}, that is, each robot $r$ can observe positions and colors of robots within a fixed distance $\phi$ ($\phi>0$ but $\phi\neq\infty$) from its current position. Since robots are identical, they share the same $\phi$.
Each robot distinguishes clockwise and counterclockwise directions according to its own \textit{chirality}.
The robots agree on a common clockwise direction if and only if they agree on a common chirality.

Each robot executes an algorithm by repeating three-phase cycles: Look, Compute, and Move phases. During the \textit{Look} phase, the robot takes a snapshot of positions and colors of robots within distance $\phi$. During the \textit{Compute} phase, the robot computes its next color and movement according to the observation in the Look phase. The robot may change its color at the end of the Compute phase. If the robot decides to move, it moves instantaneously to a neighboring node during the \textit{Move} phase. 
To model asynchrony of executions, we introduce the notion of \textit{scheduler} that decides when each robot executes phases. When the scheduler makes robot $r$ execute some phase, we say the scheduler activates the phase of $r$ or simply activates $r$. We consider three types of synchronicity: the FSYNC (fully synchronous) model, the SSYNC (semi-synchronous) model, and the ASYNC (asynchronous) model. 
In all models, time is represented by an infinite sequence of instants $0,1,2,...$.
No robot has access to this global time.
In the FSYNC and SSYNC models, all the robots that are activated at an instant $t$ execute a full cycle synchronously and concurrently between $t$ and $t+1$.
In the FSYNC model, at every instant, the scheduler activates all robots. 
In the SSYNC model, at every instant, the scheduler selects a non-empty subset of robots and activates the selected robots.
In the ASYNC model, the scheduler activates cycles of robots asynchronously: the time between Look, Compute, and Move phases is finite but unpredictable.
Note that in the ASYNC model, a robot $r$ can move based on the outdated view obtained during the previous Look phase. Throughout the paper we assume that the scheduler is \textit{fair}, that is, each robot is activated infinitely often. 

\subsection{Configuration, view, and algorithm}

\paragraph{Configuration.}
A configuration represents positions and colors of all robots.
At instant $t$, let $Q(t)$ be the set of occupied nodes, and let $M_{i, j}(t)$ be the multiset of colors of robots on node $v_{i, j} \in Q(t)$.
A \textit{configuration} $C(t)$ of the system at instant $t$ is defined as $C(t)=\left\{\left(v_{i,j},M_{i,j}(t)\right) \mid v_{i,j}\in Q(t)\right\}$.
If $t$ is clear from the context, we simply write $Q$, $M_{i,j}$ and $C$ instead of $Q(t)$, $M_{i,j}(t)$, and $C(t)$, respectively.

\paragraph{View.}
When a robot takes a snapshot of its environment, it gets a \textit{view} up to distance $\phi$.
Consider a robot $r$ on node $v_{i, j}$. 
Let $c_r$ be a color of $r$.
We describe $M_{i', j'}=\bot$ if node $v_{i',j'}$ does not exist, that is, $i'\notin\{0,1,\ldots,m-1\}$ or $j'\notin\{0,1,\ldots,n-1\}$ holds.
Since $r$ does not know the global direction, it obtains one of the following four views in case of $\phi=1$ and a common chirality:
\begin{itemize}
  \item North view: $\V_{1,\nu} = (c_r,M_{i-1,j},M_{i,j-1},M_{i,j},M_{i,j+1},M_{i+1,j})$
  \item East view: $\V_{1,e} = (c_r,M_{i,j+1},M_{i-1,j},M_{i,j},M_{i+1,j},M_{i,j-1})$
  \item South view: $\V_{1,s} = (c_r,M_{i+1,j},M_{i,j+1},M_{i,j},M_{i,j-1},M_{i-1,j})$
  \item West view: $\V_{1,w} = (c_r,M_{i,j-1},M_{i+1,j},M_{i,j},M_{i-1,j},M_{i,j+1})$
\end{itemize}
In case of $\phi=1$ and no common chirality, $r$ obtains one of eight views, which include the above four views and the mirror images of them:
\begin{itemize}
  \item Mirror image of $\V_{1,\nu}$: \\$\V_{1,\nu,\mu} = (c_r,M_{i-1,j},M_{i,j+1},M_{i,j},M_{i,j-1},M_{i+1,j})$
  \item Mirror image of $\V_{1,e}$: \\$\V_{1,e,\mu} = (c_r,M_{i,j+1},M_{i+1,j},M_{i,j},M_{i-1,j},M_{i,j-1})$
  \item Mirror image of $\V_{1,s}$: \\$\V_{1,s,\mu} = (c_r,M_{i+1,j},M_{i,j-1},M_{i,j},M_{i,j+1},M_{i-1,j})$
  \item Mirror image of $\V_{1,w}$: \\$\V_{1,w,\mu} = (c_r,M_{i,j-1},M_{i-1,j},M_{i,j},M_{i+1,j},M_{i,j+1})$
\end{itemize}
When $r$ obtains one of the views, it cannot recognize which view it obtains, however it can compute other views by rotating and/or flipping the view.
Hence, we assume that, in case of a common chirality, $r$ obtains four views $\V_{1,\nu}, \V_{1,e}, \V_{1,s}, \V_{1,w}$ when it takes a snapshot.
Note that $r$ does not recognize which view corresponds to each of North, East, South, and West views.
Similarly, we assume that, in case of no common chirality, $r$ obtains eight views $\V_{1,\nu}, \V_{1,e}, \V_{1,s}, \V_{1,w}, V_{1,\nu,\mu}, \V_{1,e,\mu}, \V_{1,s,\mu}, \V_{1,w,\mu}$ when it takes a snapshot.

Similarly, in case of $\phi=2$ and a common chirality, $r$ obtains the following four views.
\begin{itemize}
  \item North view: $\V_{2,\nu}=(c_r,M_{i-2,j},M_{i-1,j-1},M_{i-1,j},M_{i-1,j+1},M_{i,j-2},M_{i,j-1},M_{i,j},M_{i,j+1},\\M_{i,j+2},M_{i+1,j-1},M_{i+1,j},M_{i+1,j+1},M_{i+2,j})$
  \item East view: $\V_{2,e}=(c_r,M_{i,j+2},M_{i-1,j+1},M_{i,j+1},M_{i+1,j+1},M_{i-2,j},M_{i-1,j},M_{i,j},M_{i+1,j},\\M_{i+2,j},M_{i-1,j-1},M_{i,j-1},M_{i+1,j-1},M_{i,j-2})$
  \item South view: $\V_{2,s}=(c_r,M_{i+2,j},M_{i+1,j+1},M_{i+1,j},M_{i+1,j-1},M_{i,j+2},M_{i,j+1},M_{i,j},M_{i,j-1},\\M_{i,j-2},M_{i-1,j+1},M_{i-1,j},M_{i-1,j-1},M_{i-2,j})$
  \item West view: $\V_{2,w}=(c_r,M_{i,j-2},M_{i+1,j-1},M_{i,j-1},M_{i-1,j-1},M_{i+2,j},M_{i+1,j},M_{i,j},M_{i-1,j},\\M_{i-2,j},M_{i+1,j+1},M_{i,j+1},M_{i-1,j+1},M_{i,j+2})$
\end{itemize}
In case of $\phi=2$ and no common chirality, $r$ obtains eight views, which include the above four views and the mirror images of them:
\begin{itemize}
  \item Mirror image of $\V_{2,\nu}$: $\V_{2,\nu,\mu}=(c_r,M_{i-2,j},  M_{i-1,j+1},M_{i-1,j},M_{i-1,j-1},  M_{i,j+2},M_{i,j+1},M_{i,j},\\M_{i,j-1},M_{i,j-2},  M_{i+1,j+1},M_{i+1,j},M_{i+1,j-1},  M_{i+2,j})$
  \item Mirror image of $\V_{2,e}$: $\V_{2,e,\mu}=(c_r,M_{i,j+2},  M_{i+1,j+1},M_{i,j+1},M_{i-1,j+1},  M_{i+2,j},M_{i+1,j},M_{i,j},\\M_{i-1,j},M_{i-2,j},  M_{i+1,j-1},M_{i,j-1},M_{i-1,j-1},  M_{i,j-2})$
  \item Mirror image of $\V_{2,s}$: $\V_{2,s,\mu}=(c_r,M_{i+2,j},  M_{i+1,j-1},M_{i+1,j},M_{i+1,j+1},  M_{i,j-2},M_{i,j-1},M_{i,j},\\M_{i,j+1},M_{i,j+2},  M_{i-1,j-1},M_{i-1,j},M_{i-1,j+1},  M_{i-2,j})$
  \item Mirror image of $\V_{2,w}$: $\V_{2,w,\mu}=(c_r,M_{i,j-2},  M_{i-1,j-1},M_{i,j-1},M_{i+1,j-1},  M_{i-2,j},M_{i-1,j},M_{i,j},\\M_{i+1,j},M_{i+2,j},  M_{i-1,j+1},M_{i,j+1},M_{i+1,j+1},  M_{i,j+2})$
\end{itemize}

\paragraph{Algorithm.}
An algorithm is described as a set of rules. 
Each rule is represented as a combination of a label, a guard, and an action. 
The guard represents possible views obtained by a robot.
Recall that robot $r$ obtains several views during the Look phase.
If some view of robot $r$ matches a guard in some rule, we say $r$ is enabled. 
We also say the rule with the corresponding label is enabled. 
If $r$ is enabled, $r$ can execute the corresponding action (\emph{i.e.}, change its color and/or move to its neighboring node) based on the directions of the matched view during Compute and Move phases.
If several views of $r$ match some guard or some view of $r$ matches several guards, one combination of a view and a rule is selected by the scheduler. 

\subsection{Execution and problem}

\paragraph{Execution.}
An execution from initial configuration $C_0$ is a maximal sequence of configurations $E=C_0,C_1,...,C_i,...$ such that, for any $j>0$, we have 
\textit{(i)} $C_{j-1}\neq C_j$, 
\textit{(ii)} $C_j$ is obtained from $C_{j-1}$ after some robots move or change their colors, and
\textit{(iii)} for every robot $r$ that moves or changes its color between $C_{j-1}$ and $C_j$, there exists $0\leq j^\prime < j$ 
such that $r$ takes its decision to move or change its color according to its algorithm and its view in $C_{j^\prime}$.
The term ``\textit{maximal}'' means that the execution is either infinite or ends in a \textit{terminal configuration}, \textit{i.e.}, a configuration in which no robot is enabled.

\paragraph{Problem.}
A problem $\PROB$ is defined as a set of executions: An execution $E$ solves $\PROB$ if $E\in \PROB$ holds. An algorithm $\ALG$ solves problem $\PROB$ from initial configuration $C_0$ if any execution from $C_0$ solves $\PROB$. We simply say an algorithm $\ALG$ solves problem $\PROB$ if there exists an initial configuration $C_0$ such that $\ALG$ solves $\PROB$ from $C_0$. 
In this paper, we consider the terminating exploration problem.


\begin{definition}[\textbf{Terminating exploration problem}]
The terminating exploration is defined as a set of executions $E$ such that 1) every node is visited by at least one robot in $E$ and 2) there exists a suffix of $E$ such that no robots are enabled.
\end{definition}

\subsection{Descriptions}
For simplicity, we describe a rule in an algorithm with a figure in Fig.\,\ref{ruleFig}.
Figure\,\ref{ruleFig}(a) represents a rule of an algorithm in case of $\phi=1$.
Figure\,\ref{ruleFig}(b) represents a rule in case of $\phi=2$.
Each graph in Fig.\,\ref{ruleFig} represents a guard.
The guard in Fig.\,\ref{ruleFig}(a) represents a view $\V_1=(c_r,M_{i-1,j},M_{i,j-1},M_{i,j},M_{i,j+1},M_{i+1,j})$, and similarly the guard in Fig.\,\ref{ruleFig}(b) represents a view $\V_2$. 
If $M_{i',j'}=\emptyset$ holds, we paint the corresponding node white instead of writing $\emptyset$.
If $M_{i',j'}=\bot$ holds, we paint the corresponding node black instead of writing $\bot$.
If both $\emptyset$ and $\bot$ are acceptable, we paint the corresponding node gray.
If some view of robot $r$ with visible distance $\phi$ matches $\V_\phi$, $r$ is enabled.
In this case, if the scheduler activates $r$, it executes an action represented by $c_{new},Movement$.
Notation $c_{new}$ represents a new color of the robot.
Notation $Movement$ can be $Idle$, $\LEFT$, $\RIGHT$, $\UP$, $\DOWN$ and represents the movement: $Idle$ implies a robot does not move, and $\LEFT$ (resp., $\RIGHT$, $\UP$, $\DOWN$) implies a robot moves toward the node corresponding to $M_{i,j-1}$ (resp., $M_{i,j+1}$, $M_{i-1,j}$, $M_{i+1,j}$) of the guard.

\begin{figure}[t!]
\begin{center}
\includegraphics[scale=1.0]{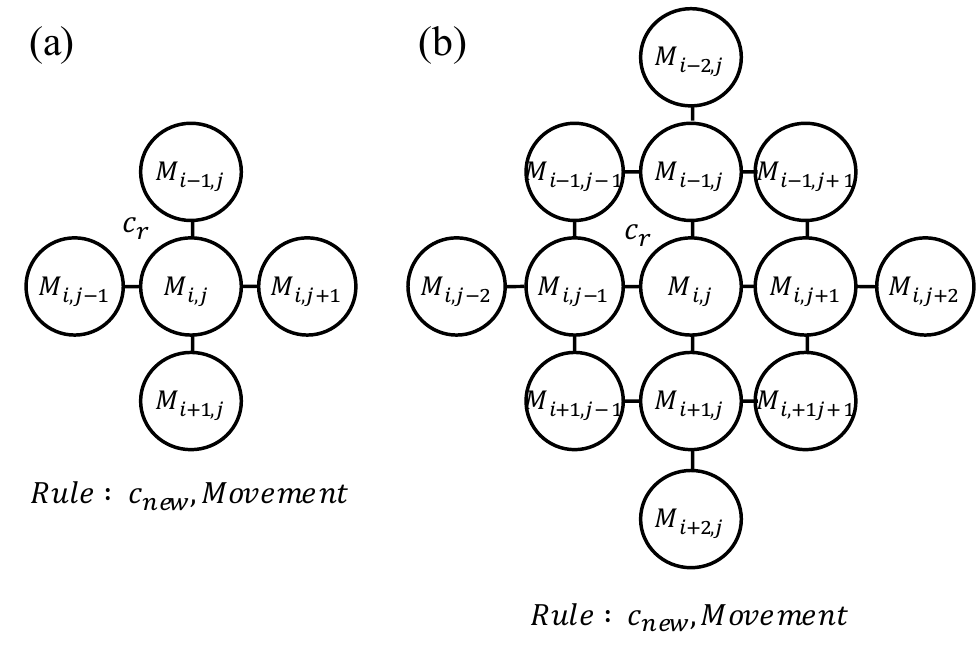}
\end{center}
\caption{Description of a rule in an algorithm}
\label{ruleFig}
\end{figure}

\section{An Impossibility result}
\label{secImpSSYNC}

In this section, we prove that, in the SSYNC model, two robots cannot achieve the terminating exploration if $\phi=1$ holds.
Since executions in the SSYNC model can happen in the ASYNC model, this impossibility also holds in the ASYNC model.
This implies that, in case of $\phi=1$, at least three robots are necessary to achieve the terminating exploration of grids in the SSYNC and ASYNC models.
In the following, we use terms of end nodes and inner nodes.
We say node $v$ is an end node if the degree of $v$ is smaller than four.
We say node $v$ is an inner node if the distance from $v$ to every end node is at least three.

\begin{theorem}
\label{grid-imp-ssync}
In case of $\phi=1$ and $k=2$, no algorithm solves the terminating exploration of grids in the SSYNC model. This holds regardless of the number of colors and a common chirality.
\end{theorem}

\begin{proof}
For contradiction, we assume that such an algorithm $\ALG$ exists.
Consider an execution $E = C_0,C_1,...$ of $\ALG$ in a $m\times n$ grid $G$ that satisfies $m\ge 9$ and $n\ge 9$.
Let $i$ be the minimum index such that some robot occupies an inner node at $C_i$.
Let $r_1$ be a robot that occupies an inner node at $C_i$ and $r_2$ be another robot.
Let $d$ be the distance between $r_1$ and $r_2$ at $C_i$.
We consider two cases: (1) $d \geq 2$ and (2) $d\leq 1$.

Consider Case 1, that is, $d \geq 2$ holds.
Let $v_1$ and $v_2$ be nodes that host $r_1$ and $r_2$, respectively, at $C_i$.
We further consider two sub-cases: (1-1) $v_2$ is not an end node, and (1-2) $v_2$ is an end node.
First assume that $v_2$ is not an end node (Case 1-1).
In this case, we can define nodes $v'_1$ and $v'_2$ such that $v'_1$ is a neighbor of $v_1$, $v'_2$ is a neighbor of $v_2$, $v'_2$ is not an end node, the distance between nodes $w_1$ and $w_2$ is at least two for any $w_1\in\{v_1,v'_1\}$ and any $w_2\in\{v_2,v'_2\}$. 
Then we can prove that the scheduler makes $r_1$ and $r_2$ stay on nodes in $\{v_1,v'_1\}$ and $\{v_2,v'_2\}$, respectively, forever after $C_i$.
Consider configuration $C$ such that $r_1$ and $r_2$ stay on nodes in $\{v_1,v'_1\}$ and $\{v_2,v'_2\}$, respectively.
Since $r_1$ and $r_2$ cannot observe each other and they are not on end nodes, $r_x$ ($x\in\{1,2\}$) cannot distinguish directions, that is, $r_x$ obtains four identical views when it takes a snapshot.
This implies that, when $r_x$ moves, the scheduler can decide which direction $r_x$ moves toward.
Hence, if $r_1$ moves, the scheduler can move $r_1$ to another node in $\{v_1,v'_1\}$.
Similarly, if $r_2$ moves, the scheduler can move $r_2$ to another node in $\{v_2,v'_2\}$.
This implies that, at the configuration after $C$, $r_1$ and $r_2$ stay on nodes in $\{v_1,v'_1\}$ and $\{v_2,v'_2\}$, respectively.
Hence, inductively, after $C_i$, robots $r_1$ and $r_2$ continue to stay on nodes in $\{v_1,v'_1\}$ and $\{v_2,v'_2\}$, respectively.
This means that robots can visit at most two inner nodes until $C_i$ and visit at most two other inner nodes after $C_i$.
Since the number of inner nodes in $G$ is at least nine, robots cannot achieve the terminating exploration.
Next assume that $v_2$ is an end node (Case 1-2).
Let $v'_1$ be an inner node that is a neighbor of $v_1$.
Similarly to Case 1-1, we can prove that, if $r_1$ never observes $r_2$, $r_1$ continues to stay on nodes in $\{v_1,v'_1\}$.
This implies that, to achieve the terminating exploration, $r_2$ moves toward $r_1$ or visits the remaining nodes by itself.
In any case, $r_2$ leaves from end nodes, which reduces to Case 1-1.

Consider Case 2, that is, $d\leq 1$ holds.
Let $v_1$ be a node that hosts $r_1$.
Let $v_2$ be a node that hosts $r_2$ if $d=1$, and a neighbor of $v_1$ if $d=0$.
We can prove that, as long as each robot moves toward another robot or stays on its current node, robots continue to stay on nodes in $\{v_1,v_2\}$: if two robots stay on different nodes, they can only move toward another node, and if two robots stay on a single node $v_1$ or $v_2$, the scheduler can move them to another node in $\{v_1,v_2\}$.
Hence, eventually a robot moves to another node, say $v_3$, when the distance between two robots is one.
In this moment, the scheduler activates only this robot.
After the movement, the distance between $r_1$ and $r_2$ is two.
Similarly to Case 1, after the configuration, robots can visit only two other inner nodes.
This implies that robots can visit at most two inner nodes ($v_1$ and $v_2$) until $C_i$ and visit at most three other inner nodes ($v_3$ and two other inner nodes) after $C_i$.
Since the number of inner nodes in $G$ is at least nine, they cannot achieve the terminating exploration.

This is a contradiction.
\end{proof}

Note that this impossibility result also holds for the perpetual exploration because the proof of Theorem \ref{grid-imp-ssync} shows that robots cannot visit some nodes in this case.

\section{Terminating Grid Exploration Algorithms}

\subsection{Overview}

In this subsection, we give the overview of our algorithms.
All of our algorithms make robots explore the grid according to the arrow in Fig.\,\ref{routeFig}.
In other words, robots start exploration from the northwest corner and repeat the following behaviors:
\begin{enumerate}
  \item Proceed east: Robots go straight to the east end of the grid.
  \item Turn west: They go one step south and turn west.
  \item Proceed west: Robots go straight to the west end of the grid.
  \item Turn east: They go one step south and turn east.
\end{enumerate}

In each algorithm, we implement the behaviors of proceeding and turning.
While proceeding, robots recognize their forward direction by their form.
In the FSYNC model, since all robots are activated at every instant, they move forward at every instant and keep their initial form.
The robots repeat this behavior until they reach the end of the grid.
On the other hand, in the SSYNC and ASYNC models, not all robots are activated at the same time.
For this reason, we propose the way to make robots move forward by moving a single robot at every instant.

The difficult part is to implement the behaviors of turning. 
Since robots do not know global directions, they must understand the south direction from the local information.
We realize this in two different approaches.
The first approach is to keep robots in two rows when proceeding east or west.
By making different forms in north and south rows, robots distinguish the two directions.
Mainly we use this approach in the case of no common chirality.
The second approach is used only in the case of a common chirality.
In this approach, robots change their form of proceeding depending on the directions.
That is, robots distinguish the east and west directions by their form.
In the case of a common chirality, robots can go south by turning right (resp. left) when they proceed east (resp. west).
In the second approach, robots do not have to keep themselves in two rows when proceeding.
This is the main reason why we can reduce the number of robots in the case of a common chirality.



\begin{figure}[t!]
\begin{center}
\includegraphics[width=10cm]{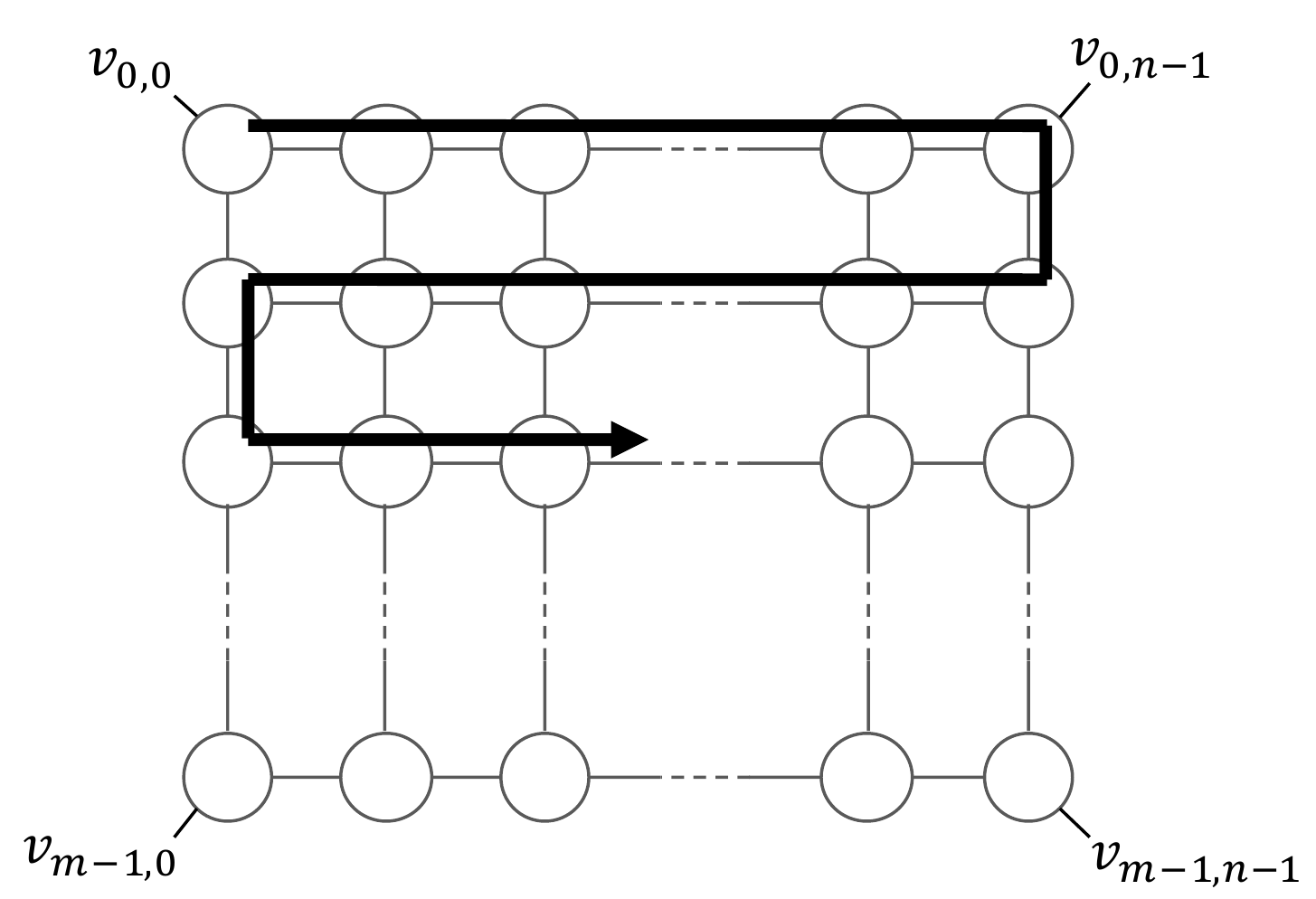}
\end{center}
\caption{Route of grid exploration with our algorithm}
\label{routeFig}
\end{figure}

In the following subsections, we give terminating grid exploration algorithms in various assumptions.
We explain a set of rules and an execution from an initial configuration with figures.
In the explanations, we mention rules that can be applied in each configuration.
We omit explanations why other rules cannot be applied, but readers can easily check it by comparing the configuration and the set of rules.

\subsection{Algorithms for the FSYNC model}
In this subsection, we give terminating grid exploration algorithms for the FSYNC model.

\subsubsection{$\phi=2$, $\ell=2$, a common chirality, and $k=2$}
\label{secF22T2}
We give a terminating exploration algorithm for $m\times n$ grids $(m\geq2, n\geq3)$ in case of $\phi=2$, $\ell=2$, a common chirality, and $k=2$.
A set of colors is $Col=\{\G, \W\}$.
The algorithm is given in Algorithm \ref{algorithmF22T2}.

\begin{algorithm}[t!]
\caption{Fully Synchronous Terminating Exploration for $\phi=2,\,\ell=2,$ $k=2$ with a Common Chirality}
\label{algorithmF22T2}
\begin{algorithmic}
\renewcommand{\algorithmicrequire}{\textbf{Initial configuration}}
\REQUIRE
\STATE $\{(v_{0,0},\{\G\}),(v_{0,1},\{\W\})\}$
\renewcommand{\algorithmicrequire}{\textbf{Rules}}
\REQUIRE
\STATE
  \centering
  \includegraphics[width=0.95\textwidth]{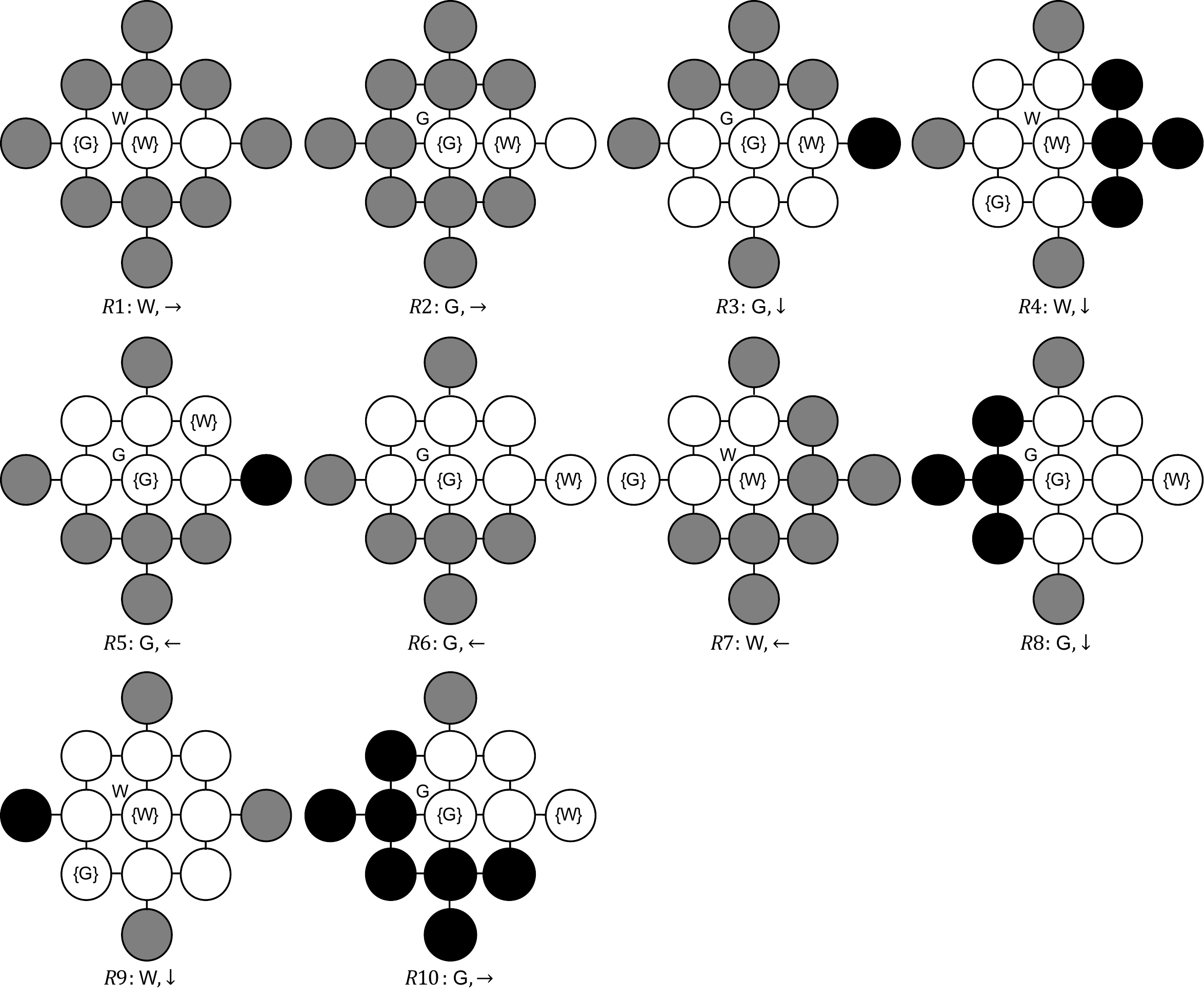}
\end{algorithmic}
\end{algorithm}

\paragraph{Proceeding east.}
From the initial configuration, robots with color $\G$ and $\W$ can execute rules $R1$ and $R2$, respectively.
Hence, they proceed east while keeping the form.

\paragraph{Turning west.}
The process of turning west is shown in Fig.\,\ref{turnWestF22T2}.
\begin{figure}[tbp]
\begin{center}
 \includegraphics[scale=0.8]{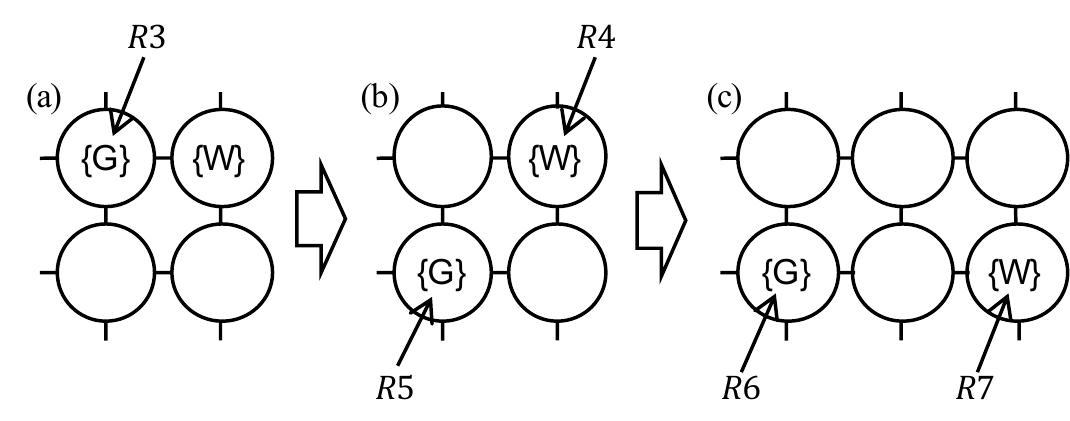}
\end{center}
\caption{Turning west in an execution of Algorithm\,\ref{algorithmF22T2}}
\label{turnWestF22T2}
\end{figure}
After robots proceed east, they reach the east end of the grid (Fig.\,\ref{turnWestF22T2}(a)). 
From this configuration, the robot with color $\G$ moves south by rule $R3$, and hence the configuration becomes one in Fig.\,\ref{turnWestF22T2}(b).
From this configuration, the robot with color $\W$ moves south by rule $R4$.
At the same time, the robot with color $\G$ moves west by rule $R5$.
Hence, the configuration becomes one in Fig.\,\ref{turnWestF22T2}(c).

\paragraph{Proceeding west.}
From the configuration in Fig.\,\ref{turnWestF22T2}(c), the robot with color $\G$ and the robot with color $\W$ can execute rules $R6$ and $R7$, respectively.
Hence, they proceed west while keeping the form.

\paragraph{Turning east.}
The process of turning east is shown in Fig.\,\ref{turnEastF22T2}.
\begin{figure}[tbp]
\begin{center}
 \includegraphics[scale=0.8]{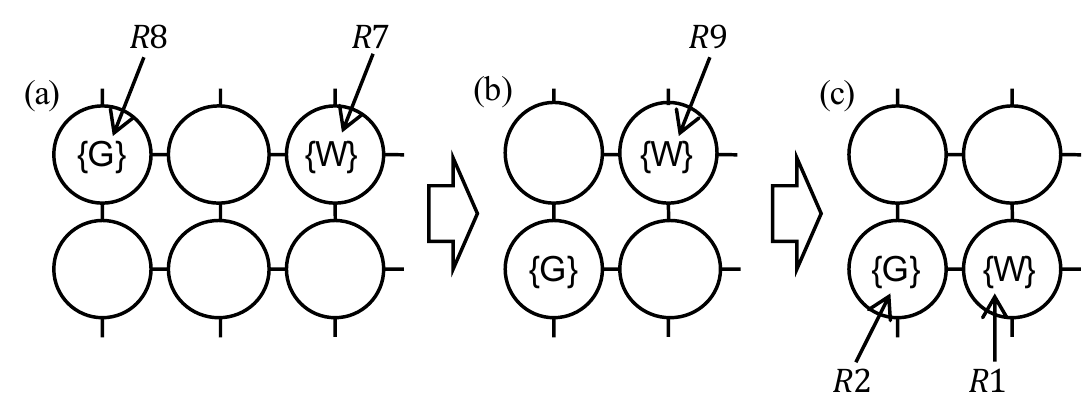}
\end{center}
\caption{Turning east in an execution of Algorithm\,\ref{algorithmF22T2}}
\label{turnEastF22T2}
\end{figure}
After robots proceed west, they reach the west end of the grid (Fig.\,\ref{turnEastF22T2}(a)). 
From this configuration, the robot with color $\G$ moves south by rule $R8$.
At the same time, the robot with color $\W$ moves by rule $R7$.
Hence, the configuration becomes one in Fig.\,\ref{turnEastF22T2}(b).
From this configuration, the robot with color $\W$ moves south by rule $R9$, and hence the configuration becomes one in Fig.\,\ref{turnEastF22T2}(c).
From this configuration, two robots can proceed east again.

\paragraph{End of exploration.}
After robots visit all nodes and reach a south corner of the grid, the configuration becomes terminal.
In case that $m$ is odd, two robots visit the south end nodes while proceeding east, and hence they reach the southeast corner.
Immediately after node $v_{m-1, n-1}$ is visited, the configuration is $\{(v_{m-1,n-2},\{\G\}),(v_{m-1,n-1},\{\W\})\}$.
At this configuration, no robots are enabled.
In case that $m$ is even, two robots visit the south end nodes while proceeding west, and hence they reach the southwest corner.
Immediately after node $v_{m-1, 0}$ is visited, the configuration is $\{(v_{m-1,0},\{\G\}),(v_{m-1,2},\{\W\})\}$.
From this configuration, robots with colors $\G$ and $\W$ move by rules $R10$ and $R7$, respectively.
Hence, the configuration becomes $\{(v_{m-1,1},\{\G,\W\})\}$.
At this configuration, no robots are enabled.

\subsubsection{$\phi=2$, $\ell=2$, no common chirality, and $k=3$}
\label{secF22F3}
We give a terminating exploration algorithm for $m\times n$ grids $(m\geq2, n\geq3)$ in case of $\phi=2$, $\ell=2$, no common chirality, and $k=3$.
A set of colors is $Col=\{\G, \W\}$.
The algorithm is given in Algorithm \ref{algorithmF22F3}.

\begin{algorithm}[tbp]
\caption{Fully Synchronous Terminating Exploration for $\phi=2,\,\ell=2,$ $k=3$ Without Common Chirality}
\label{algorithmF22F3}
\begin{algorithmic}
\renewcommand{\algorithmicrequire}{\textbf{Initial configuration}}
\REQUIRE
\STATE $\{(v_{0,0},\{\G\}),(v_{0,1},\{\G\}),(v_{1,0},\{\W\})\}$
\renewcommand{\algorithmicrequire}{\textbf{Rules}}
\REQUIRE
\STATE
  \centering
  \includegraphics[width=0.95\textwidth]{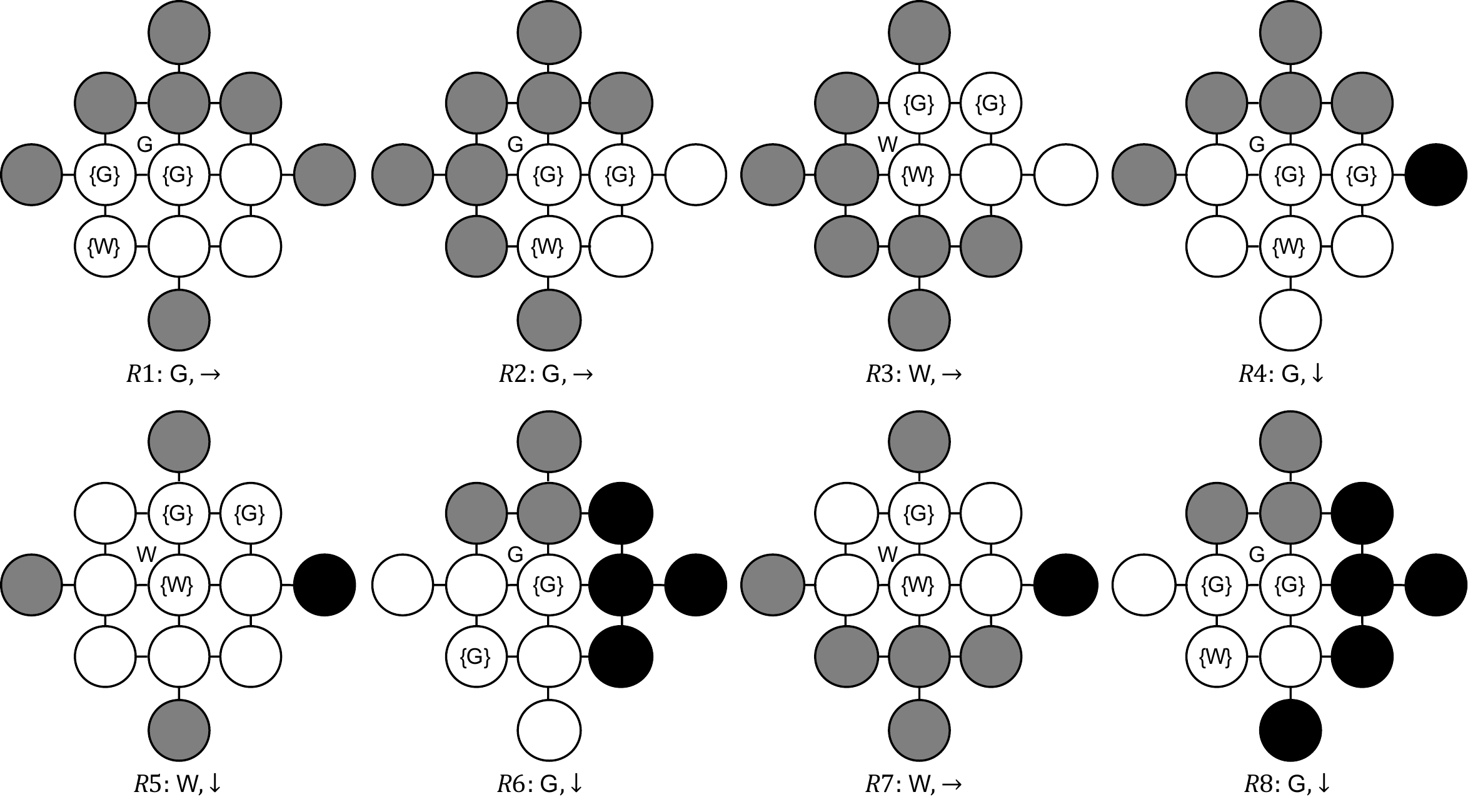}
\end{algorithmic}
\end{algorithm}

\paragraph{Proceeding east.}
At the initial configuration, the robot on $v_{0,1}$ can execute rule $R1$, the robot on $v_{0,0}$ can execute rule $R2$, and the robot on $v_{1,0}$ can execute rule $R3$.
By repeatedly executing those rules, robots proceed east while keeping the form.

\paragraph{Turning west.}
The process of turning west is shown in Fig.\,\ref{turnWestF22F3}.
\begin{figure}[tbp]
\begin{center}
 \includegraphics[scale=0.8]{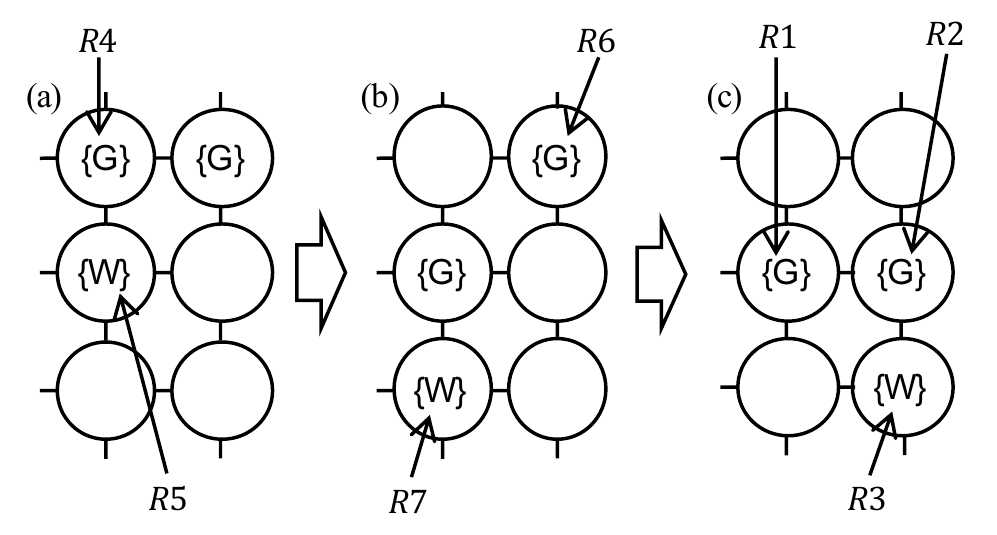}
\end{center}
\caption{Turning west in an execution of Algorithm\,\ref{algorithmF22F3}}
\label{turnWestF22F3}
\end{figure}
After robots proceed east, they reach the east end of the grid (Fig.\,\ref{turnWestF22F3}(a)). 
From this configuration, two robots on west nodes move south by rules $R4$ and $R5$.
Hence, the configuration becomes one in Fig.\,\ref{turnWestF22F3}(b).
From this configuration, the robot with color $\G$ at the east end of the grid moves south by rule $R6$ and the robot with color $\W$ moves east by rule $R7$.
Consequently, the configuration becomes one in Fig.\,\ref{turnWestF22F3}(c).

\paragraph{Proceeding west and turning east.}
The form of robots in Fig.\,\ref{turnWestF22F3}(c) is a mirror image of the one that robots make to proceed east.
Hence, robots proceed west and turn east with the same rules as proceeding east and turning west, respectively.

\paragraph{End of exploration.}
In case that $m$ is odd, robots visit the south end nodes while proceeding west.
Eventually, the configuration becomes $\{(v_{m-2,0},\{\G\}),(v_{m-2,1},\{\G\}),(v_{m-1,1},\{\W\})\}$.
Node $v_{m-1,0}$ has not been visited yet.
From this configuration, the robot on $v_{m-2,0}$ moves to $v_{m-1,0}$ by rule $R8$, and hence the configuration becomes $\{(v_{m-1,0},\{\G\}),(v_{m-2,1},\{\G\}),(v_{m-1,1},\{\W\})\}$.
At this configuration, no robots are enabled.
In case that $m$ is even, robots terminate the algorithm similarly to the odd case.

\subsubsection{$\phi=2$, $\ell=1$, a common chirality, and $k=3$}
\label{secF21T3}
In executions of Algorithm \ref{algorithmF22T2}, robots do not change their colors and robots with different colors do not occupy a single node.
Therefore, by representing the robot of color $\W$ in Algorithm \ref{algorithmF22T2} with two robots of color $\G$, we can construct a terminating exploration algorithm in case of $\phi=2$, $\ell=1$, a common chirality, and $k=3$.

\subsubsection{$\phi=2$, $\ell=1$, no common chirality, and $k=4$}
\label{secF21F4}
In executions of Algorithm \ref{algorithmF22F3}, robots do not change their colors and robots with different colors do not occupy a single node.
Therefore, by representing the robot of color $\W$ in Algorithm \ref{algorithmF22F3} with two robots of color $\G$, we can construct a terminating exploration algorithm in case of $\phi=2$, $\ell=1$, no common chirality, and $k=4$.

\subsubsection{$\phi=1$, $\ell=3$, a common chirality, and $k=2$}
\label{secF13T2}
We give a terminating exploration algorithm for $m\times n$ grids $(m\geq2, n\geq3)$ in case of $\phi=1$, $\ell=3$, a common chirality, and $k=2$.
A set of colors is $Col=\{\G, \W, \B\}$.
The algorithm is given in Algorithm \ref{algorithmF13T2}.

\begin{algorithm}[tbp]
\caption{Fully Synchronous Terminating Exploration for $\phi=1,\,\ell=3,$ $k=2$ with Common Chirality}
\label{algorithmF13T2}
\begin{algorithmic}
\renewcommand{\algorithmicrequire}{\textbf{Initial configuration}}
\REQUIRE
\STATE $\{(v_{0,0},\{\G\}),(v_{0,1},\{\W\})\}$
\renewcommand{\algorithmicrequire}{\textbf{Rules}}
\REQUIRE
\STATE
  \centering
  \includegraphics[width=0.95\textwidth]{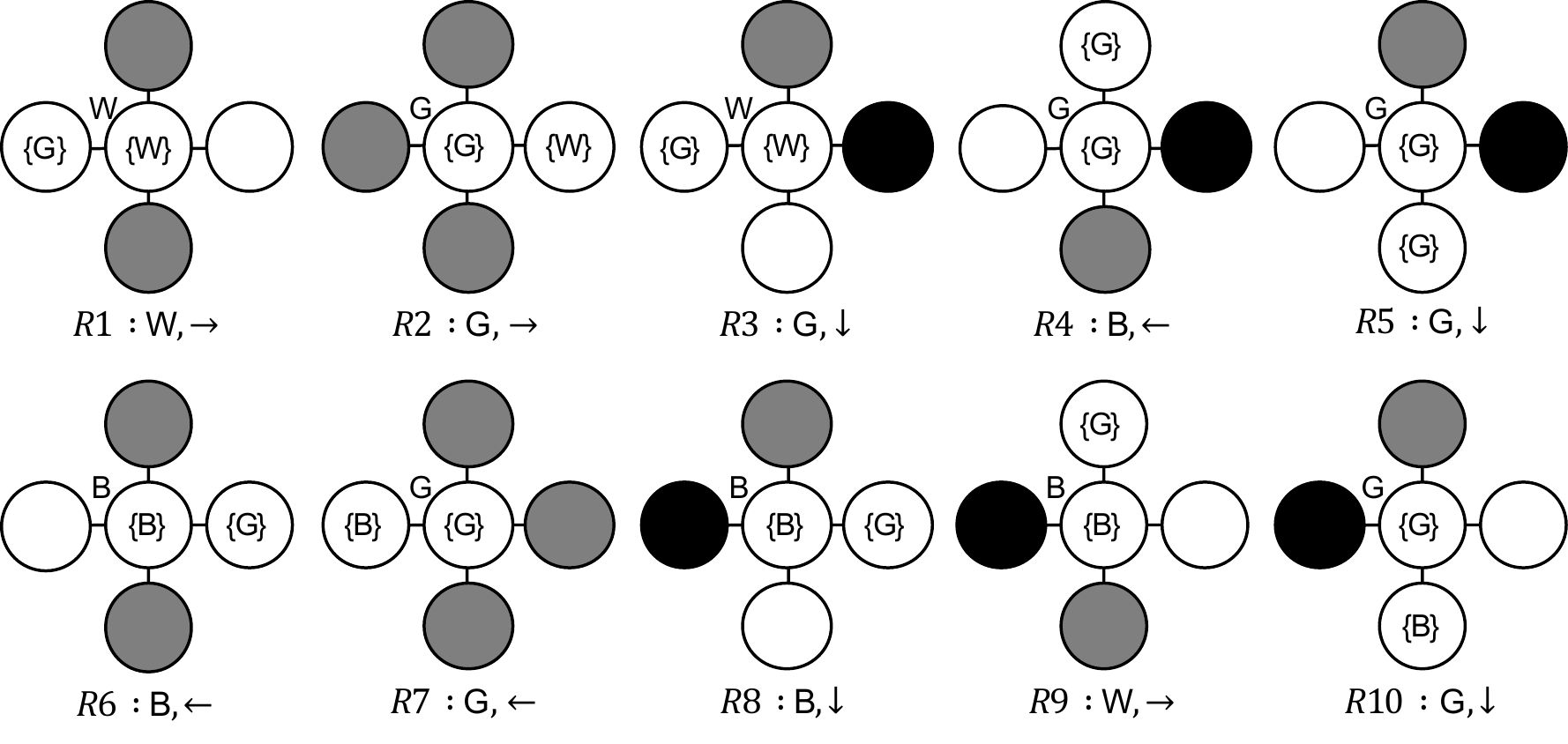}
\end{algorithmic}
\end{algorithm}

\paragraph{Proceeding east.}
From the initial configuration, robots with colors $\W$ and $\G$ can execute rules $R1$ and $R2$, respectively.
Hence, they proceed east while keeping the form.

\paragraph{Turning west.}
The process of turning west is shown in Fig.\,\ref{turnWestF13T2}.
\begin{figure}[tbp]
\begin{center}
 \includegraphics[scale=0.8]{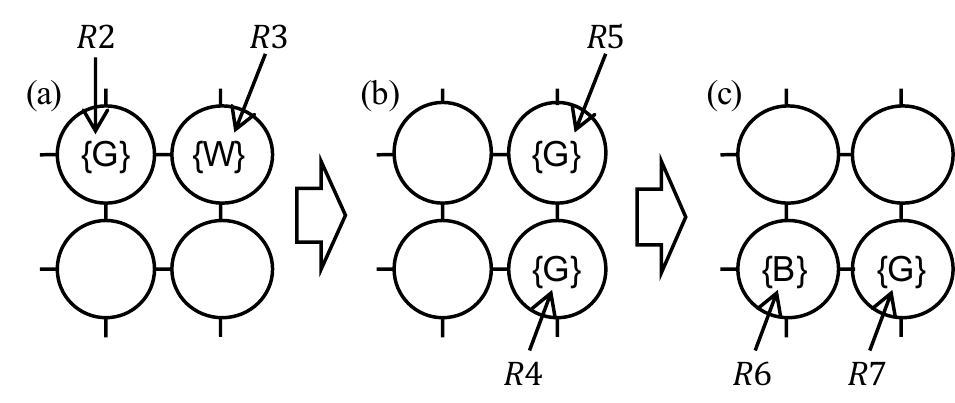}
\end{center}
\caption{Turning west in an execution of Algorithm\,\ref{algorithmF13T2}}
\label{turnWestF13T2}
\end{figure}
After robots proceed east, they reach the east end of the grid (Fig.\,\ref{turnWestF13T2}(a)). 
From this configuration, the robot with color $\W$ moves south by rule $R3$.
At the same time, the robot with color $\G$ moves east by rule $R2$.
Hence, the configuration becomes one in Fig.\,\ref{turnWestF13T2}(b).
From this configuration, the robot on a south node changes its color to $\B$ and moves west by rule $R4$.
At the same time, the robot on a north node moves south by rule $R5$.
Consequently, the configuration becomes one in Fig.\,\ref{turnWestF13T2}(c).

\paragraph{Proceeding west.}
From the configuration in Fig.\,\ref{turnWestF13T2}(c), the robot with color $\B$ and the robot with color $\G$ can execute rules $R6$ and $R7$, respectively.
Hence, they proceed west while keeping the form.

\paragraph{Turning east.}
The process of turning east is shown in Fig.\,\ref{turnEastF13T2}.
\begin{figure}[tbp]
\begin{center}
 \includegraphics[scale=0.8]{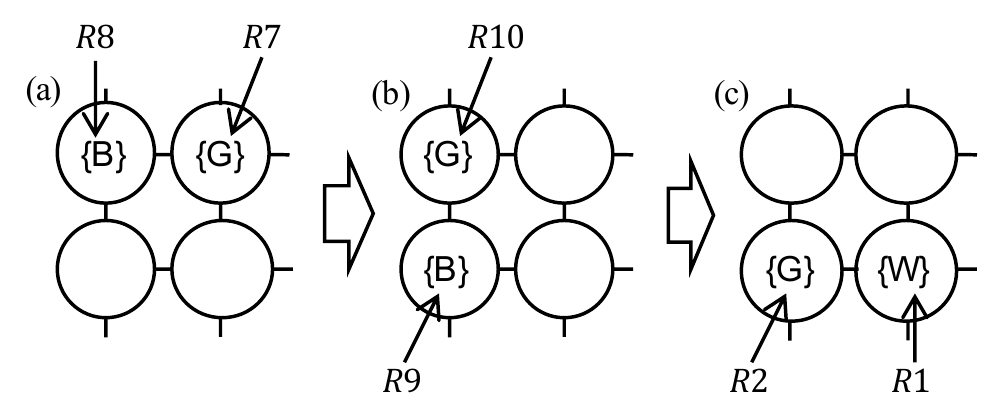}
\end{center}
\caption{Turning east in an execution of Algorithm\,\ref{algorithmF13T2}}
\label{turnEastF13T2}
\end{figure}
After robots proceed west, they reach the west end of the grid (Fig.\,\ref{turnEastF13T2}(a)). 
From this configuration, the robot with color $\B$ moves south by rule $R8$.
At the same time, the robot with color $\G$ moves west by rule $R7$.
Hence, the configuration becomes one in Fig.\,\ref{turnEastF13T2}(b).
From this configuration, the robot with color $\B$ changes its color to $\W$ and moves east by rule $R9$.
At the same time, the robot with color $\G$ moves south by rule $R10$, and hence the configuration becomes one in Fig.\,\ref{turnEastF13T2}(c).
From this configuration, two robots can proceed east again.

\paragraph{End of exploration.}
In case that $m$ is odd, two robots visit the south end nodes while proceeding east, and hence they reach the southeast corner.
Immediately after node $v_{m-1, n-1}$ is visited, the configuration is $\{(v_{m-1,n-2},\{\G\}),(v_{m-1,n-1},\{\W\})\}$.
From this configuration, the robot with color $\G$ moves, and hence the configuration becomes $\{(v_{m-1,n-1},\{\G,\W\})\}$.
At this configuration, no robots are enabled.
In case that $m$ is even, two robots visit the south end nodes while proceeding west, and hence they reach the southwest corner.
Immediately after node $v_{m-1, 0}$ is visited, the configuration is $\{(v_{m-1,0},\{\B\}),(v_{m-1,1},\{\G\})\}$.
From this configuration, the robot with color $\G$ moves by rule $R7$, and hence the configuration becomes $\{(v_{m-1,0},\{\G,\B\})\}$.
At this configuration, no robots are enabled.

\subsubsection{$\phi=1$, $\ell=3$, no common chirality, and $k=4$}
\label{secF13F4}
We give a terminating exploration algorithm for $m\times n$ grids $(m\geq2, n\geq3)$ in case of $\phi=1$, $\ell=3$, no common chirality, and $k=4$.
A set of colors is $Col=\{\G, \W, \B\}$.
The algorithm is given in Algorithm \ref{algorithmF13F4}.

\begin{algorithm}[tbp]
\caption{Fully Synchronous Terminating Exploration for $\phi=1,\,\ell=3,$ $k=4$ Without Common Chirality}
\label{algorithmF13F4}
\begin{algorithmic}
\renewcommand{\algorithmicrequire}{\textbf{Initial configuration}}
\REQUIRE
\STATE $\{(v_{0,0},\{\G\}),(v_{0,1},\{\W\}),(v_{1,0},\{\B\}),(v_{1,1},\{\W\})\}$
\renewcommand{\algorithmicrequire}{\textbf{Rules}}
\REQUIRE
\STATE
  \centering
  \includegraphics[width=0.95\textwidth]{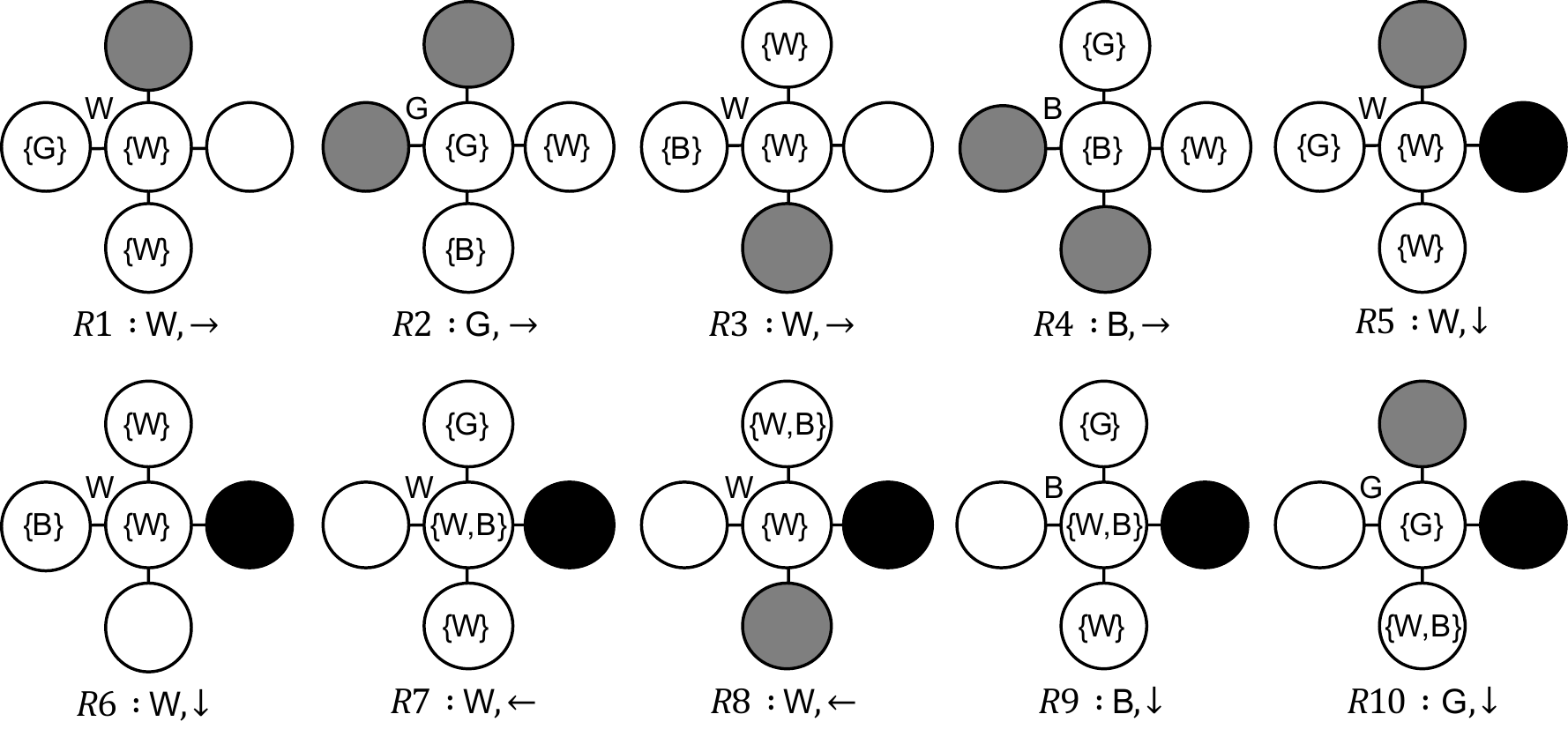}
\end{algorithmic}
\end{algorithm}

\paragraph{Proceeding east.}
At the initial configuration, the robot on $v_{0,1}$ can execute rule $R1$, the robot on $v_{0,0}$ can execute rule $R2$, the robot on $v_{1,1}$ can execute rule $R3$, and the robot on $v_{1,0}$ can execute rule $R4$.
By repeatedly executing those rules, robots proceed east while keeping the form.

\paragraph{Turning west.}
The process of turning west is shown in Fig.\,\ref{turnWestF13F4}.
\begin{figure}[tbp]
\begin{center}
 \includegraphics[scale=0.8]{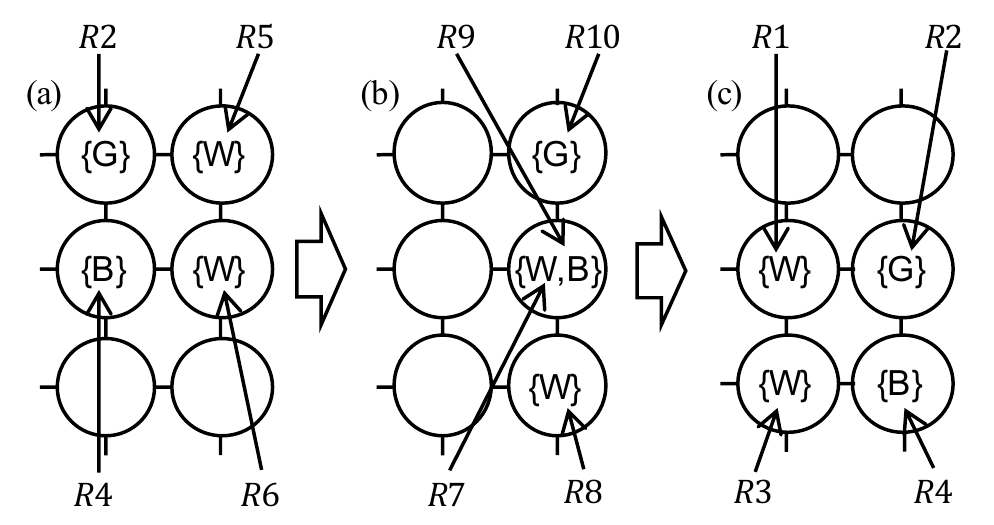}
\end{center}
\caption{Turning west in an execution of Algorithm\,\ref{algorithmF13F4}}
\label{turnWestF13F4}
\end{figure}
After robots proceed east, they reach the east end of the grid (Fig.\,\ref{turnWestF13F4}(a)). 
From this configuration, two robots on east nodes move south by rules $R5$ and $R6$.
At the same time, the other robots move east by rules $R2$ and $R4$.
Hence, the configuration becomes one in Fig.\,\ref{turnWestF13F4}(b).
From this configuration, two robots with color $\W$ move west by rules $R7$ and $R8$.
At the same time, robots with color $\B$ and $\G$ move south by rules $R9$ and $R10$, respectively.
Consequently, the configuration becomes one in Fig.\,\ref{turnWestF13F4}(c).

\paragraph{Proceeding west and turning east.}
The form of robots in Fig.\,\ref{turnWestF13F4}(c) is a mirror image of the one that robots make to proceed east.
Hence, robots proceed west and turn east with the same rules as proceeding east and turning west, respectively.

\paragraph{End of exploration.}
In case that $m$ is odd, robots visit the south end nodes while proceeding west, and hence they reach the southwest corner.
Immediately after node $v_{m-1, 0}$ is visited, the configuration is $\{(v_{m-2,0},\{\W\}),(v_{m-2,1},\{\G\}),(v_{m-1,0},\{\W\}),(v_{m-1,1},\{\B\})\}$.
From this configuration, the robot on $v_{m-2,0}$ moves to $v_{m-1,0}$ by rule $R5$.
At the same time, robots with colors $\G$ and $\B$ move west by rules $R2$ and $R4$, respectively.
Hence, the configuration becomes $\{(v_{m-2,0},\{\G\}),(v_{m-1,0},\{\W,\W,\B\})\}$.
At this configuration, no robots are enabled.
In case that $m$ is even, robots terminate the algorithm similarly to the odd case.

\subsubsection{$\phi=1$, $\ell=2$, a common chirality, and $k=3$}
\label{secF12T3}
We give a terminating exploration algorithm for $m\times n$ grids $(m\geq2, n\geq3)$ in case of $\phi=1$, $\ell=2$, a common chirality, and $k=3$.
A set of colors is $Col=\{\G, \W, \B\}$.
The algorithm is given in Algorithm \ref{algorithmF12T3}.

\begin{algorithm}[tbp]
\caption{Fully Synchronous Terminating Exploration for $\phi=1,\,\ell=2,$ $k=3$ with Common Chirality}
\label{algorithmF12T3}
\begin{algorithmic}
\renewcommand{\algorithmicrequire}{\textbf{Initial configuration}}
\REQUIRE
\STATE $\{(v_{0,0},\{\G\}),(v_{0,1},\{\G\}),(v_{1,0},\{\W\})\}$
\renewcommand{\algorithmicrequire}{\textbf{Rules}}
\REQUIRE
\STATE
  \centering
  \includegraphics[width=0.95\textwidth]{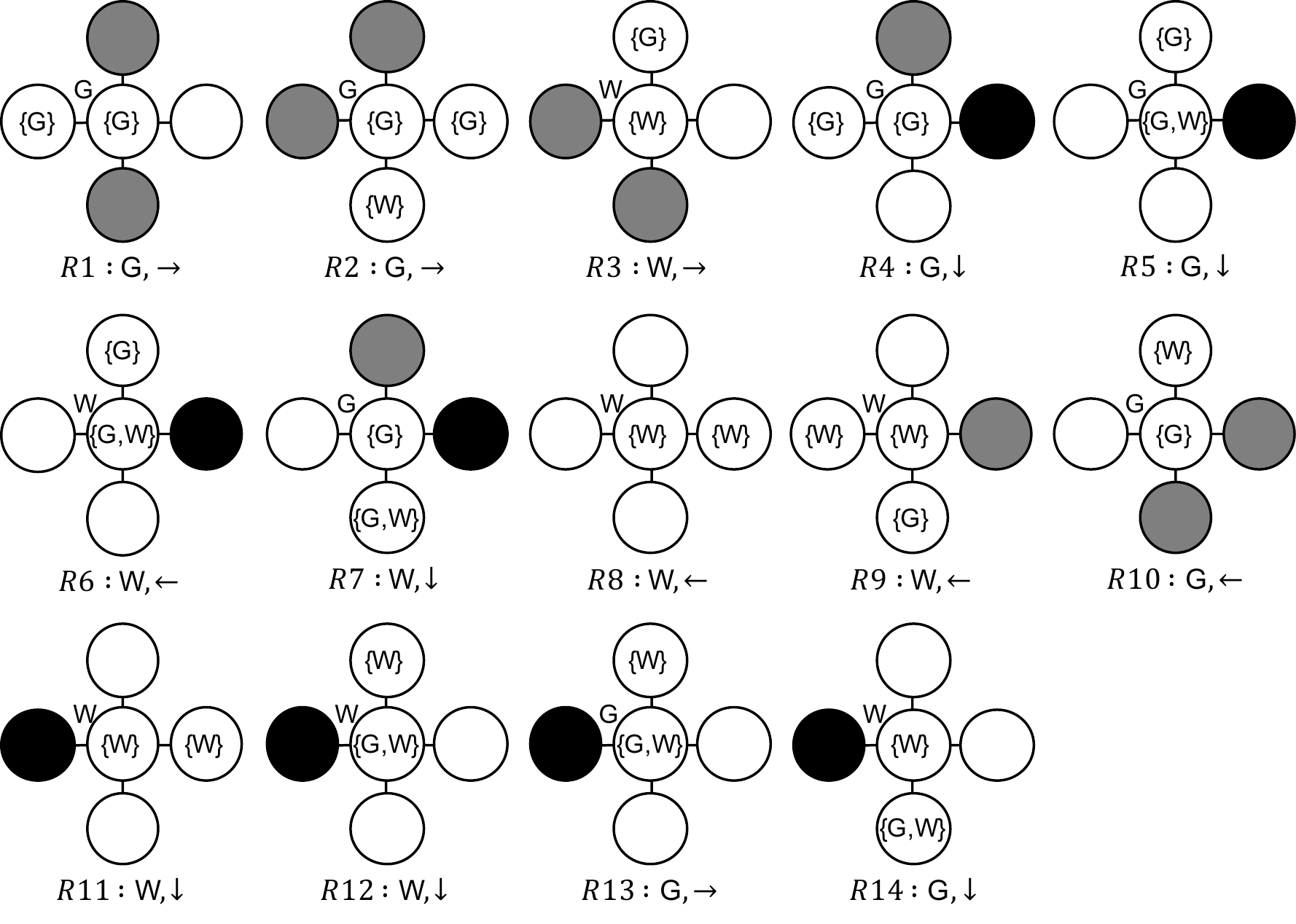}
\end{algorithmic}
\end{algorithm}

\paragraph{Proceeding east.}
At the initial configuration, the robot on $v_{0,1}$ can execute rule $R1$, the robot on $v_{0,0}$ can execute rule $R2$, and the robot on $v_{1,0}$ can execute rule $R3$.
By repeatedly executing those rules, robots proceed east while keeping the form.

\paragraph{Turning west.}
The process of turning west is shown in Fig.\,\ref{turnWestF12T3}.
\begin{figure}[tbp]
\begin{center}
 \includegraphics[scale=0.8]{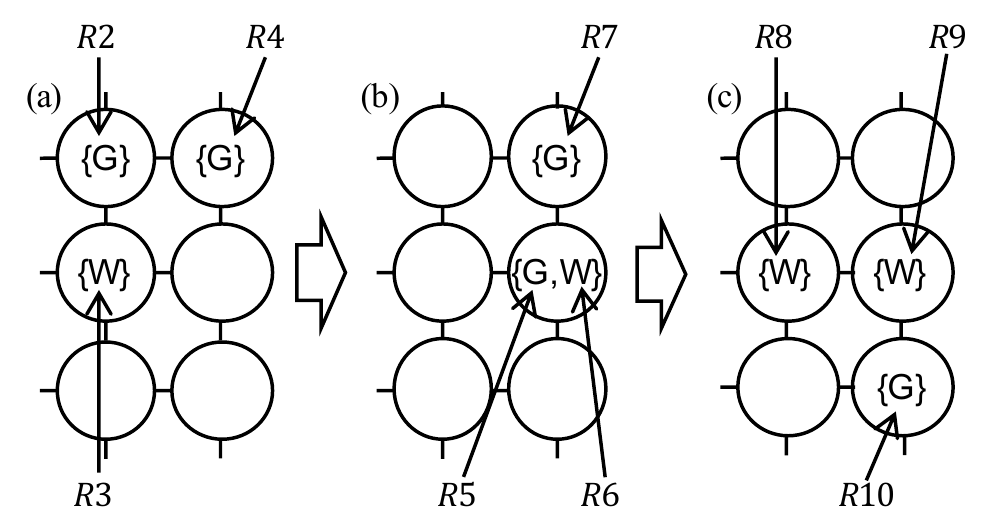}
\end{center}
\caption{Turning west in an execution of Algorithm\,\ref{algorithmF12T3}}
\label{turnWestF12T3}
\end{figure}
After robots proceed east, they reach the east end of the grid (Fig.\,\ref{turnWestF12T3}(a)). 
From this configuration, the robot at the east end moves south by rule $R4$.
At the same time, the other robots move east by rules $R2$ and $R3$.
Hence, the configuration becomes one in Fig.\,\ref{turnWestF12T3}(b).
From this configuration, the robot with color $\G$ on a south node moves south by rule $R5$.
At the same time, the robot with color $\W$ moves west by rule $R6$, and the robot on a north node changes its color to $\W$ and moves south.
Consequently, the configuration becomes one in Fig.\,\ref{turnWestF12T3}(c).

\paragraph{Proceeding west.}
At the configuration in Fig.\,\ref{turnWestF12T3}(c), the robot on a west node can execute rule $R8$, the robot with color $\W$ on a east node can execute rule $R9$, and the robot with color $\G$ can execute rule $R10$.
Hence, they proceed west while keeping the form.

\paragraph{Turning east.}
The process of turning east is shown in Fig.\,\ref{turnEastF12T3}.
\begin{figure}[tbp]
\begin{center}
 \includegraphics[scale=0.8]{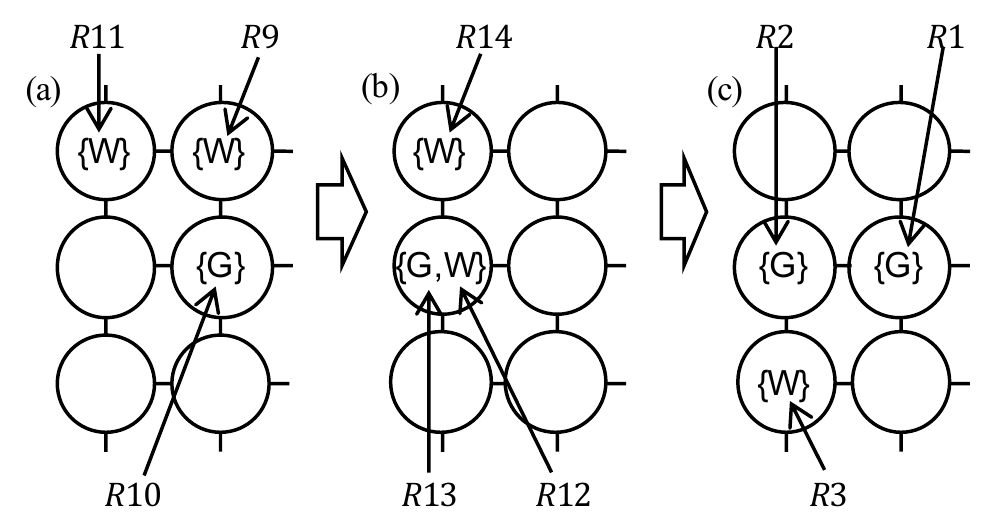}
\end{center}
\caption{Turning east in an execution of Algorithm\,\ref{algorithmF12T3}}
\label{turnEastF12T3}
\end{figure}
After robots proceed west, they reach the west end of the grid (Fig.\,\ref{turnEastF12T3}(a)). 
From this configuration, the robot on a west node moves south by rule $R11$.
At the same time, the other robots move west by rule $R9$ and $R10$.
Hence, the configuration becomes one in Fig.\,\ref{turnEastF12T3}(b).
From this configuration, the robot with color $\W$ on a south node moves south by rule $R12$.
At the same time, the robot with color $\G$ moves east rule $R13$, and the robot on a north node changes its color to $\G$ and moves south by rule $R14$.
Hence, the configuration becomes one in Fig.\,\ref{turnEastF12T3}(c).
From this configuration, three robots can proceed east again.

\paragraph{End of exploration.}
In case that $m$ is odd, robots visit the south end nodes while proceeding west.
Eventually, the configuration becomes $\{(v_{m-2,0},\{\W\}),(v_{m-2,1},\{\W\}),(v_{m-1,1},\{\G\})\}$.
Node $v_{m-1,0}$ has not been visited yet.
From this configuration, the robot on $v_{m-2,0}$ moves to $v_{m-1,0}$ by rule $R11$.
At the same time, the other robots move west by rules $R9$ and $R10$, and hence the configuration becomes $\{(v_{m-2,0},\{\W\}),(v_{m-1,0},\{\G,\W\})\}$.
From this configuration, the robot on $v_{m-2,0}$ moves to $v_{m-1,0}$ by rule $R14$, and hence the configuration becomes $\{(v_{m-1,0},\{\G,\G,\W\})\}$.
At this configuration, no robots are enabled.
In case that $m$ is even, robots visit the south end nodes while proceeding east.
Eventually, the configuration becomes $\{(v_{m-2,n-2},\{\G\}),(v_{m-2,n-1},\{\G\}),(v_{m-1,n-2},\{\W\})\}$.
Node $v_{m-1,n-1}$ has not been visited yet.
From this configuration, the robot on $v_{m-2,n-1}$ moves to $v_{m-1,n-1}$ by rule $R4$.
At the same time, the other robots move east by rules $R2$ and $R3$, and hence the configuration becomes $\{(v_{m-2,n-1},\{\G\}),(v_{m-1,n-1},\{\G,\W\})\}$.
From this configuration, the robot on $v_{m-2,n-1}$ moves to $v_{m-1,n-1}$ by rule $R7$, and hence the configuration becomes $\{(v_{m-1,n-1},\{\G,\W,\W\})\}$.
At this configuration, no robots are enabled.

\subsubsection{$\phi=1$, $\ell=2$, no common chirality, and $k=5$}
\label{secF12F5}
In executions of Algorithm \ref{algorithmF13F4}, robots do not change their colors and robots with colors $\G$ and $\B$ do not occupy a single node.
Therefore, by representing the robot of color $\B$ in Algorithm \ref{algorithmF13F4} with two robots of color $\G$, we can construct a terminating exploration algorithm in case of $\phi=1$, $\ell=2$, no common chirality, and $k=5$.


\subsection{Algorithms for the ASYNC model}
In this subsection, we give terminating exploration algorithms for the ASYNC model.
Clearly robots can achieve terminating exploration with those algorithms also in the SSYNC and FSYNC models.

\subsubsection{$\phi=2$, $\ell=3$, a common chirality, and $k=2$}
\label{secA23T2}
We give a terminating exploration algorithm for $m\times n$ grids $(m\geq2, n\geq3)$ in case of $\phi=2$, $\ell=3$, a common chirality, and $k=2$.
A set of colors is $Col=\{\G, \W, \B\}$.
The algorithm is given in Algorithm \ref{algorithmA23T2}.

\begin{algorithm}[tbp]
\caption{Asynchronous Terminating Exploration for $\phi=2,\,\ell=3,$ $k=2$ with Common Chirality}
\label{algorithmA23T2}
\begin{algorithmic}
\renewcommand{\algorithmicrequire}{\textbf{Initial configuration}}
\REQUIRE
\STATE $\{(v_{0,0},\{\G\}),(v_{0,1},\{\W\})\}$
\renewcommand{\algorithmicrequire}{\textbf{Rules}}
\REQUIRE
\STATE
  \centering
  \includegraphics[width=0.95\textwidth]{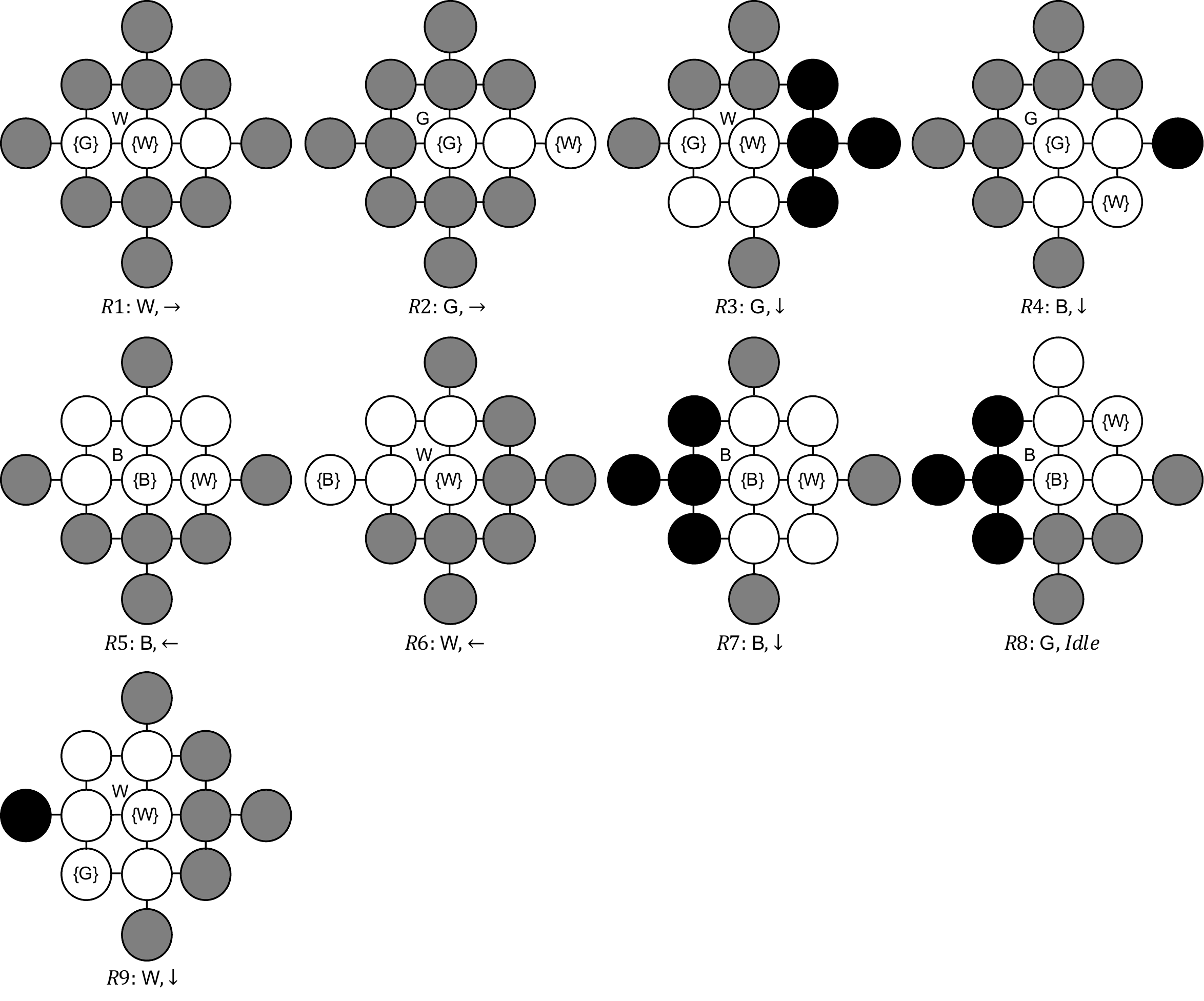}
\end{algorithmic}
\end{algorithm}

\paragraph{Proceeding east.}
From the initial configuration, the robot with color $\W$ moves east by rule $R1$, and hence the configuration becomes $\{(v_{0,0},\{\G\}),(v_{0,2},\{\W\})\}$.
From this configuration, the robot with color $\G$ moves east by rule $R2$, and hence the configuration becomes $\{(v_{0,1},\{\G\}),(v_{0,2},\{\W\})\}$.
After that, robots proceed east while keeping the form by repeatedly executing those rules. 

\paragraph{Turning west.}
The process of turning west is shown in Fig.\,\ref{turnWestA23T2}.
\begin{figure}[tbp]
\begin{center}
 \includegraphics[scale=0.8]{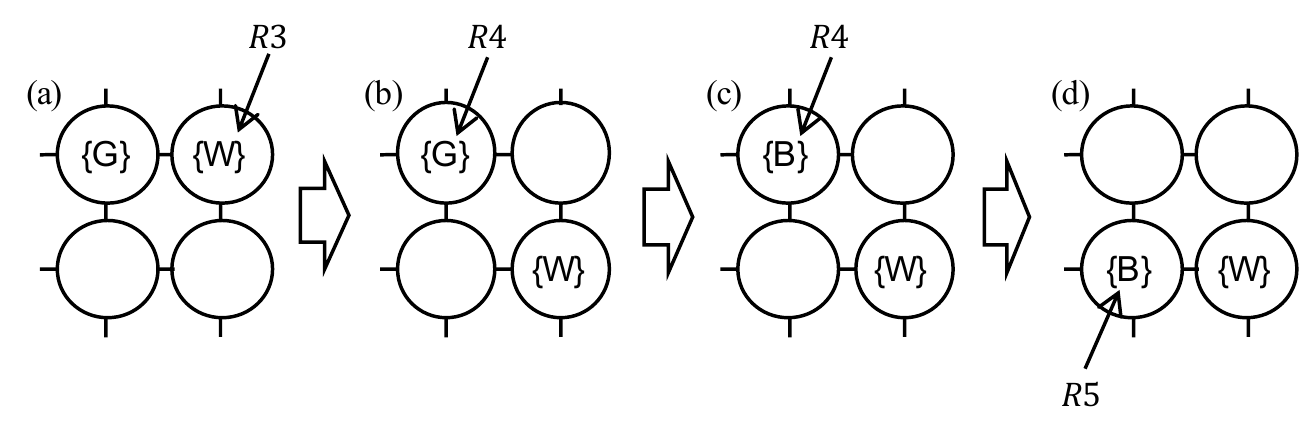}
\end{center}
\caption{Turning west in an execution of Algorithm\,\ref{algorithmA23T2}}
\label{turnWestA23T2}
\end{figure}
After robots proceed east, they reach the east end of the grid (Fig.\,\ref{turnWestA23T2}(a)). 
From this configuration, the robot with color $\W$ moves south by rule $R3$, and hence the configuration becomes one in Fig.\,\ref{turnWestA23T2}(b).
From this configuration, the robot with color $\G$ changes its color to $\B$ and moves south by rule $R4$.
In the ASYNC model, after the robot with color $\G$ changes its color, the other robot may observe the intermediate configuration (Fig.\,\ref{turnWestA23T2}(c)).
However, there are no rules that the other robot can execute in the intermediate configuration.
Consequently, the configuration becomes one in Fig.\,\ref{turnWestA23T2}(d).

\paragraph{Proceeding west.}
From the configuration in Fig.\,\ref{turnWestA23T2}(d), the robot with color $\B$ moves west by rule $R5$.
Next, the robot with color $\W$ moves west by rule $R6$.
After that, robots proceed west while keeping the form by repeatedly executing those rules. 

\paragraph{Turning east.}
The process of turning east is shown in Fig.\,\ref{turnEastA23T2}.
\begin{figure}[tbp]
\begin{center}
 \includegraphics[scale=0.8]{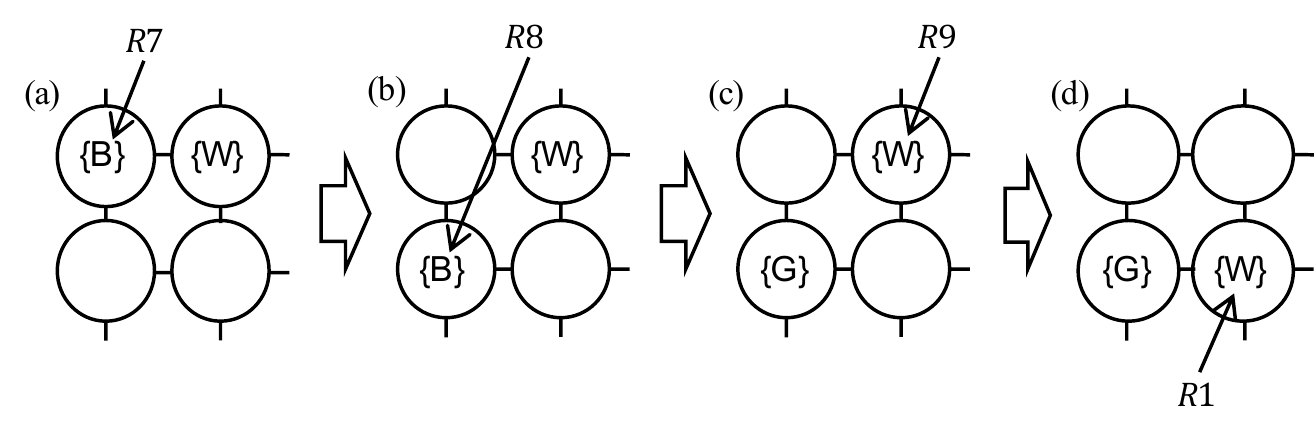}
\end{center}
\caption{Turning east in an execution of Algorithm\,\ref{algorithmA23T2}}
\label{turnEastA23T2}
\end{figure}
After robots proceed west, they reach the west end of the grid (Fig.\,\ref{turnEastA23T2}(a)). 
From this configuration, the robot with color $\B$ moves south by rule $R7$, and hence the configuration becomes one in Fig.\,\ref{turnEastA23T2}(b).
From this configuration, the robot with color $\B$ changes its color to $\G$ by rule $R8$, and hence the configuration becomes one in Fig.\,\ref{turnEastA23T2}(c).
From this configuration, the robot with color $\W$ moves south by rule $R9$, and hence the configuration becomes one in Fig.\,\ref{turnEastA23T2}(d).
From this configuration, two robots can proceed east again.

\paragraph{End of exploration.}
In case that $m$ is odd, two robots visit the south end nodes while proceeding east, and hence they reach the southeast corner.
Immediately after node $v_{m-1, n-1}$ is visited, the configuration is $\{(v_{m-1,n-2},\{\G\}),(v_{m-1,n-1},\{\W\})\}$.
At this configuration, no robots are enabled.
In case that $m$ is even, two robots visit the south end nodes while proceeding west, and hence they reach the southwest corner.
Immediately after node $v_{m-1, 0}$ is visited, the configuration is $\{(v_{m-1,0},\{\B\}),(v_{m-1,1},\{\W\})\}$.
At this configuration, no robots are enabled.

\subsubsection{$\phi=2$, $\ell=3$, no common chirality, and $k=3$}
\label{secA23F3}
We give a terminating exploration algorithm for $m\times n$ grids $(m\geq2, n\geq3)$ in case of $\phi=2$, $\ell=3$, a common chirality, and $k=2$.
A set of colors is $Col=\{\G, \W, \B\}$.
The algorithm is given in Algorithm \ref{algorithmA23F3}.

\begin{algorithm}[tbp]
\caption{Asynchronous Terminating Exploration for $\phi=2,\,\ell=3,$ $k=3$ Without Common Chirality}
\label{algorithmA23F3}
\begin{algorithmic}
\renewcommand{\algorithmicrequire}{\textbf{Initial configuration}}
\REQUIRE
\STATE $\{(v_{0,0},\{\G\}),(v_{0,1},\{\W\}),(v_{1,0},\{\B\})\}$
\renewcommand{\algorithmicrequire}{\textbf{Rules}}
\REQUIRE
\STATE
  \centering
  \includegraphics[width=0.95\textwidth]{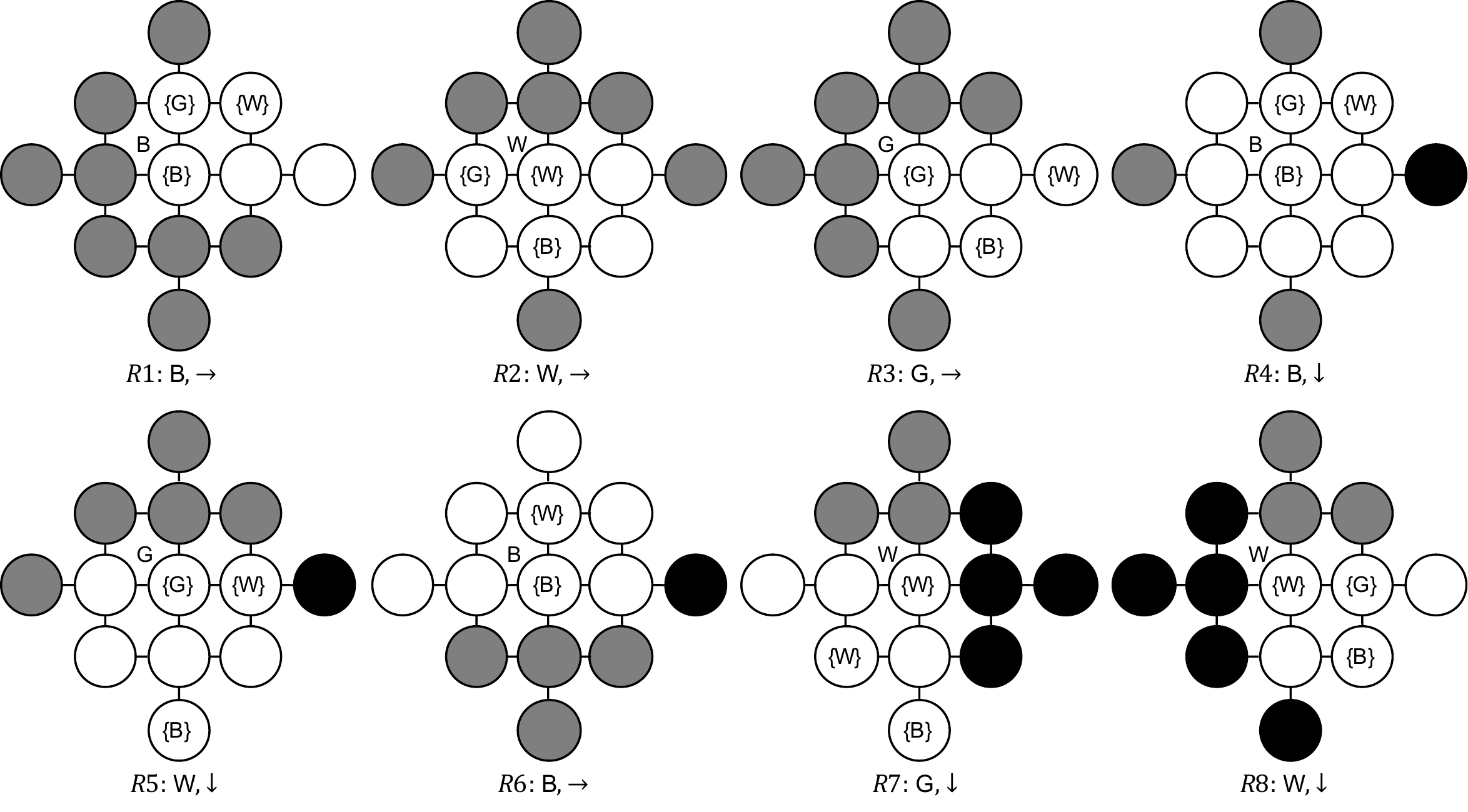}
\end{algorithmic}
\end{algorithm}

\paragraph{Proceeding east.}
From the initial configuration, the robot with color $\B$ moves by rule $R1$, and hence the configuration becomes $\{(v_{0,0},\{\G\}),(v_{0,1},\{\W\}),(v_{1,1},\{\B\})\}$.
From this configuration, the robot with color $\W$ by rule $R2$, and hence the configuration becomes $\{(v_{0,0},\{\G\}),(v_{0,2},\{\W\}),(v_{1,1},\{\B\})\}$.
From this configuration, the robot with color $\G$ by rule $R3$, and hence the configuration becomes $\{(v_{0,1},\{\G\}),(v_{0,2},\{\W\}),(v_{1,1},\{\B\})\}$.
After that, robots proceed east while keeping the form by repeatedly executing those rules. 

\paragraph{Turning west.}
The process of turning west is shown in Fig.\,\ref{turnWestA23F3}.
\begin{figure}[tbp]
\begin{center}
 \includegraphics[scale=0.8]{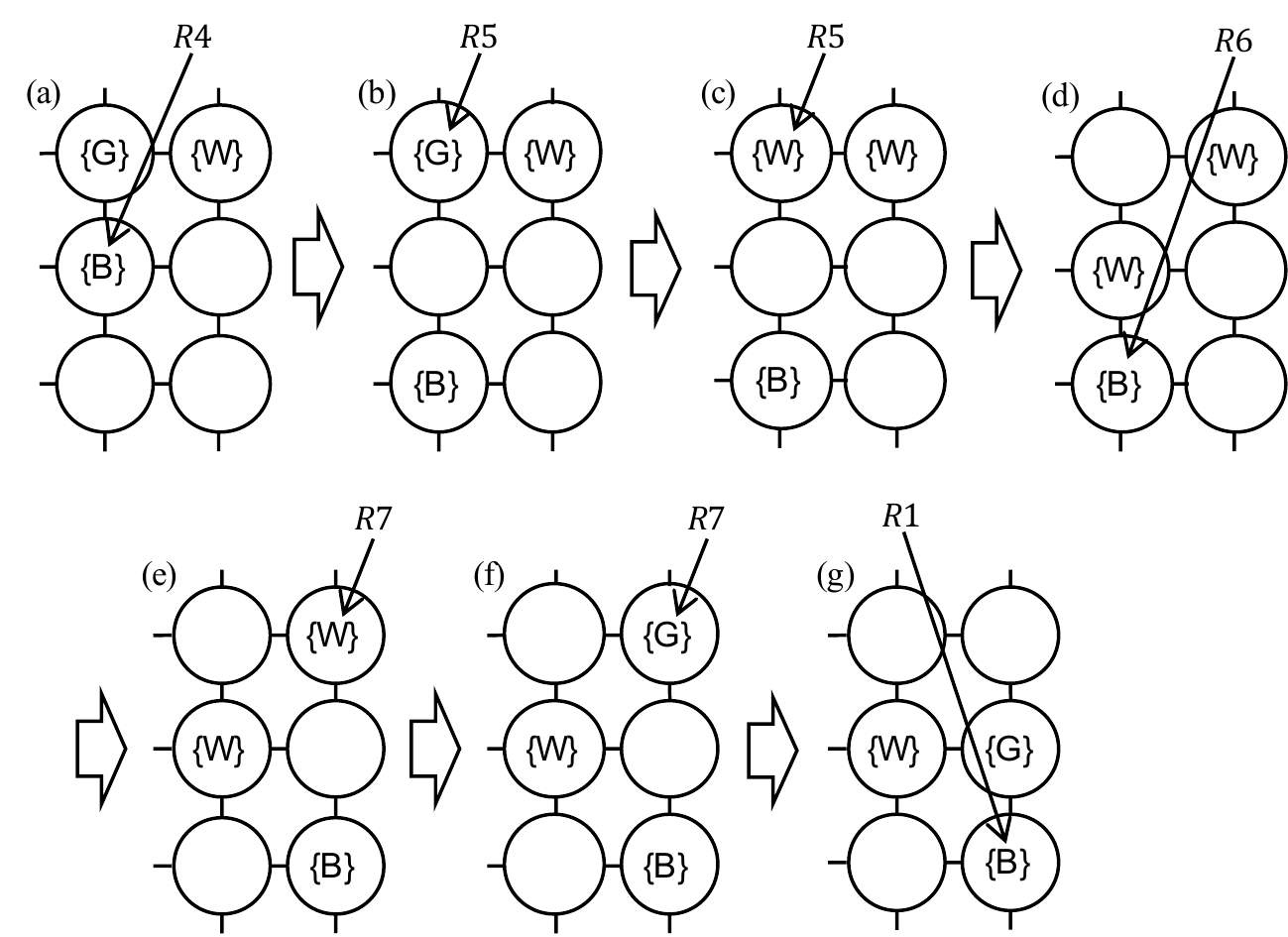}
\end{center}
\caption{Turning west in an execution of Algorithm\,\ref{algorithmA23F3}}
\label{turnWestA23F3}
\end{figure}
After robots proceed east, they reach the east end of the grid (Fig.\,\ref{turnWestA23F3}(a)). 
From this configuration, the robot with color $\B$ moves south by rule $R4$, and hence the configuration becomes one in Fig.\,\ref{turnWestA23F3}(b).
From this configuration, the robot with color $\G$ changes its color to $\W$ and moves south by rule $R5$.
In the ASYNC model, after the robot with color $\G$ changes its color, other robots may observe the intermediate configuration (Fig.\,\ref{turnWestA23F3}(c)).
However, there are no rules that the other robot can execute in the intermediate configuration.
Hence, the configuration becomes one in Fig.\,\ref{turnWestA23F3}(d).
From this configuration, the robot with color $\B$ moves east by rule $R6$, and hence the configuration becomes one in Fig.\,\ref{turnWestA23F3}(e).
From this configuration, the robot with color $\W$ changes its color to $\G$ and moves south by rule $R7$.
In the ASYNC model, after the robot with color $\W$ changes its color, other robots may observe the intermediate configuration (Fig.\,\ref{turnWestA23F3}(f)).
However, there are no rules that the other robot can execute in the intermediate configuration.
Consequently, the configuration becomes one in Fig.\,\ref{turnWestA23F3}(g).

\paragraph{Proceeding west and turning east.}
The form of robots in Fig.\,\ref{turnWestA23F3}(g) is a mirror image of the one that robots make to proceed east.
Hence, robots proceed west and turn east with the same rules as proceeding east and turning west, respectively.

\paragraph{End of exploration.}
In case that $m$ is odd, robots visit the south end nodes while proceeding west.
Eventually, the configuration becomes $\{(v_{m-2,0},\{\W\}),(v_{m-2,1},\{\G\}),(v_{m-1,1},\{\B\})\}$.
Node $v_{m-1,0}$ has not been visited yet.
From this configuration, the robot with color $\W$ moves to $v_{m-1,0}$ by rule $R8$, and hence the configuration becomes $\{(v_{m-2,1},\{\G\}),(v_{m-1,0},\{\W\}),(v_{m-1,1},\{\B\})\}$.
At this configuration, no robots are enabled.
In case that $m$ is even, robots terminate the algorithm similarly to the odd case.

\subsubsection{$\phi=2$, $\ell=2$, a common chirality, and $k=3$}
\label{secA22T3}
We give a terminating exploration algorithm for $m\times n$ grids $(m\geq2, n\geq3)$ in case of $\phi=2$, $\ell=2$, a common chirality, and $k=3$.
A set of colors is $Col=\{\G, \W\}$.
The algorithm is given in Algorithm \ref{algorithmA22T3}.

\begin{algorithm}[tbp]
\caption{Asynchronous Terminating Exploration for $\phi=2,\,\ell=2,$ $k=3$ with Common Chirality}
\label{algorithmA22T3}
\begin{algorithmic}
\renewcommand{\algorithmicrequire}{\textbf{Initial configuration}}
\REQUIRE
\STATE $\{(v_{0,0},\{\G\}),(v_{0,1},\{\W\}),(v_{1,0},\{\G\})\}$
\renewcommand{\algorithmicrequire}{\textbf{Rules}}
\REQUIRE
\STATE
  \centering
  \includegraphics[width=0.95\textwidth]{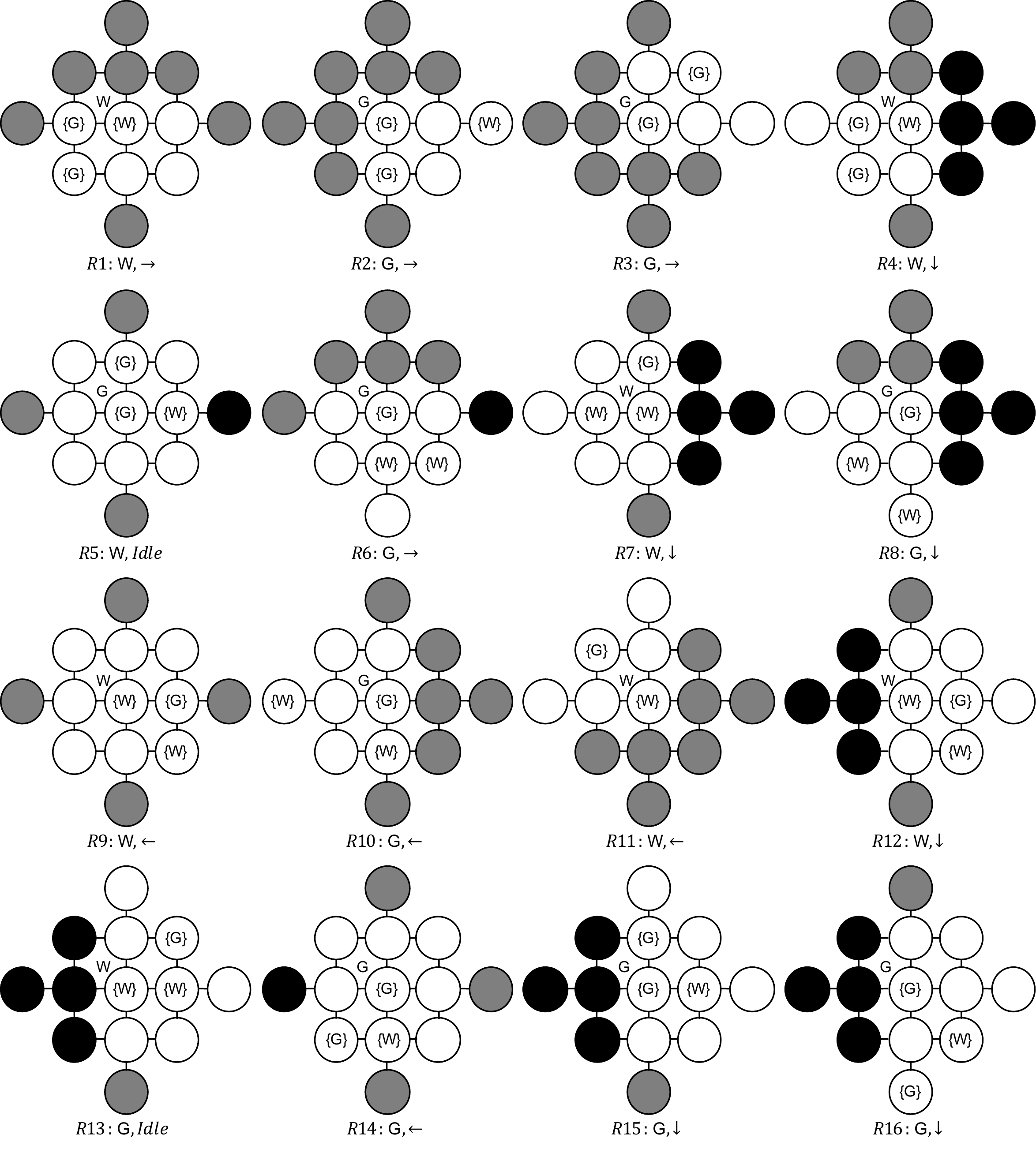}
\end{algorithmic}
\end{algorithm}

\paragraph{Proceeding east.}
From the initial configuration, the robot with color $\W$ moves east by rule $R1$, and hence the configuration becomes $\{(v_{0,0},\{\G\}),(v_{0,2},\{\W\}),(v_{1,0},\{\G\})\}$.
From this configuration, the robot on $v_{0,0}$ moves east by rule $R2$, and hence the configuration becomes $\{(v_{0,1},\{\G\}),(v_{0,2},\{\W\}),(v_{1,0},\{\G\})\}$.
From this configuration, the robot on $v_{1,0}$ moves east by rule $R3$, and hence the configuration becomes $\{(v_{0,1},\{\G\}),(v_{0,2},\{\W\}),(v_{1,1},\{\G\})\}$.
After that, robots proceed east while keeping the form by repeatedly executing those rules. 

\paragraph{Turning west.}
The process of turning west is shown in Fig.\,\ref{turnWestA22T3}.
\begin{figure}[tbp]
\begin{center}
 \includegraphics[scale=0.8]{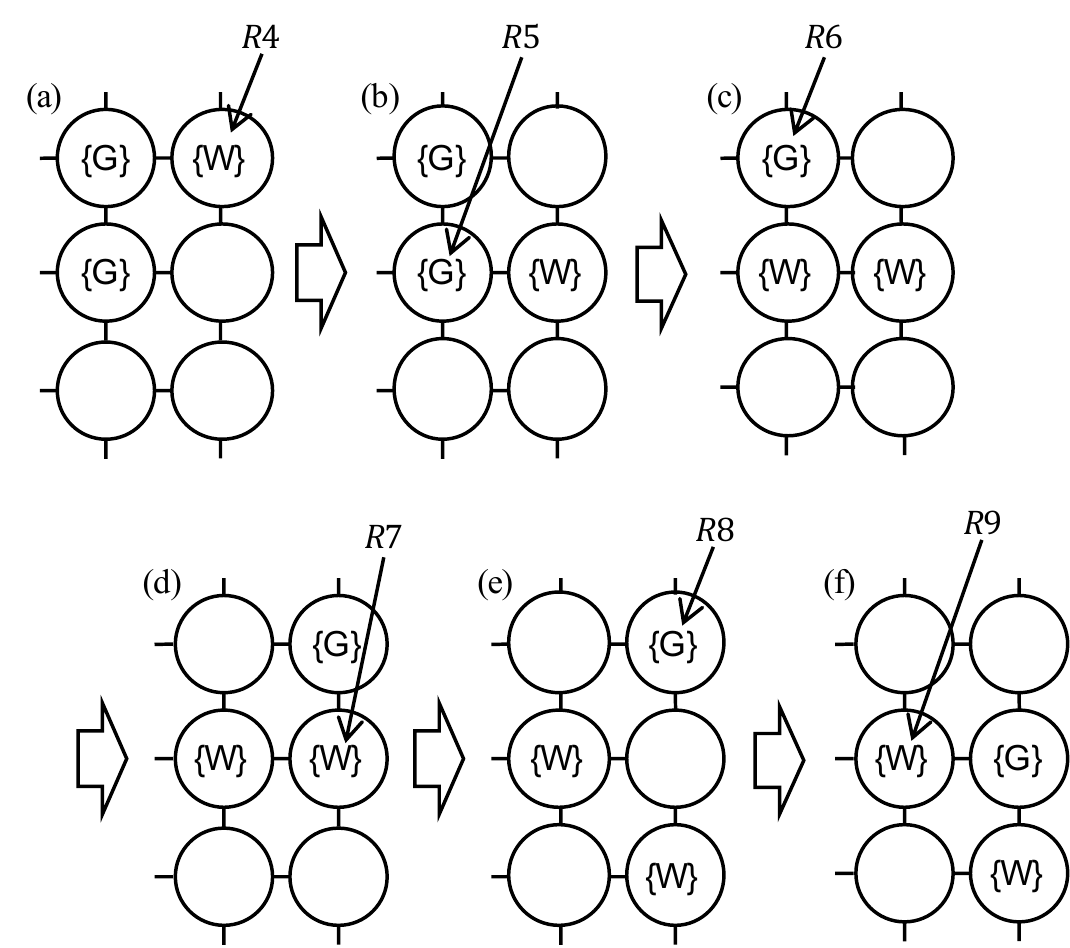}
\end{center}
\caption{Turning west in an execution of Algorithm\,\ref{algorithmA22T3}}
\label{turnWestA22T3}
\end{figure}
After robots proceed east, they reach the east end of the grid (Fig.\,\ref{turnWestA22T3}(a)). 
From this configuration, the robot with color $\W$ moves south by rule $R4$, and hence the configuration becomes one in Fig.\,\ref{turnWestA22T3}(b).
From this configuration, the robot with color $\G$ on a south node changes its color to $\W$ by rule $R5$, and hence the configuration becomes one in Fig.\,\ref{turnWestA22T3}(c).
From this configuration, the robot with color $\G$ moves east by rule $R6$, and hence the configuration becomes one in Fig.\,\ref{turnWestA22T3}(d).
From this configuration, the robot with color $\W$ moves south by rule $R7$, and hence the configuration becomes one in Fig.\,\ref{turnWestA22T3}(e).
From this configuration, the robot with color $\G$ moves south by rule $R8$, and hence the configuration becomes one in Fig.\,\ref{turnWestA22T3}(f).

\paragraph{Proceeding west.}
From the configuration in Fig.\,\ref{turnWestA22T3}(f), the robot with color $\W$ on a west node moves west by rule $R9$.
Next, the robot with color $\G$ moves west by rule $R10$.
Then, the robot with color $\W$ on a east node moves west by rule $R11$.
After that, robots proceed west while keeping the form by repeatedly executing those rules. 

\paragraph{Turning east.}
The process of turning east in an execution of Algorithm \ref{algorithmA22T3} is shown in Fig.\,\ref{turnEastA22T3}.
\begin{figure}[tbp]
\begin{center}
 \includegraphics[scale=0.8]{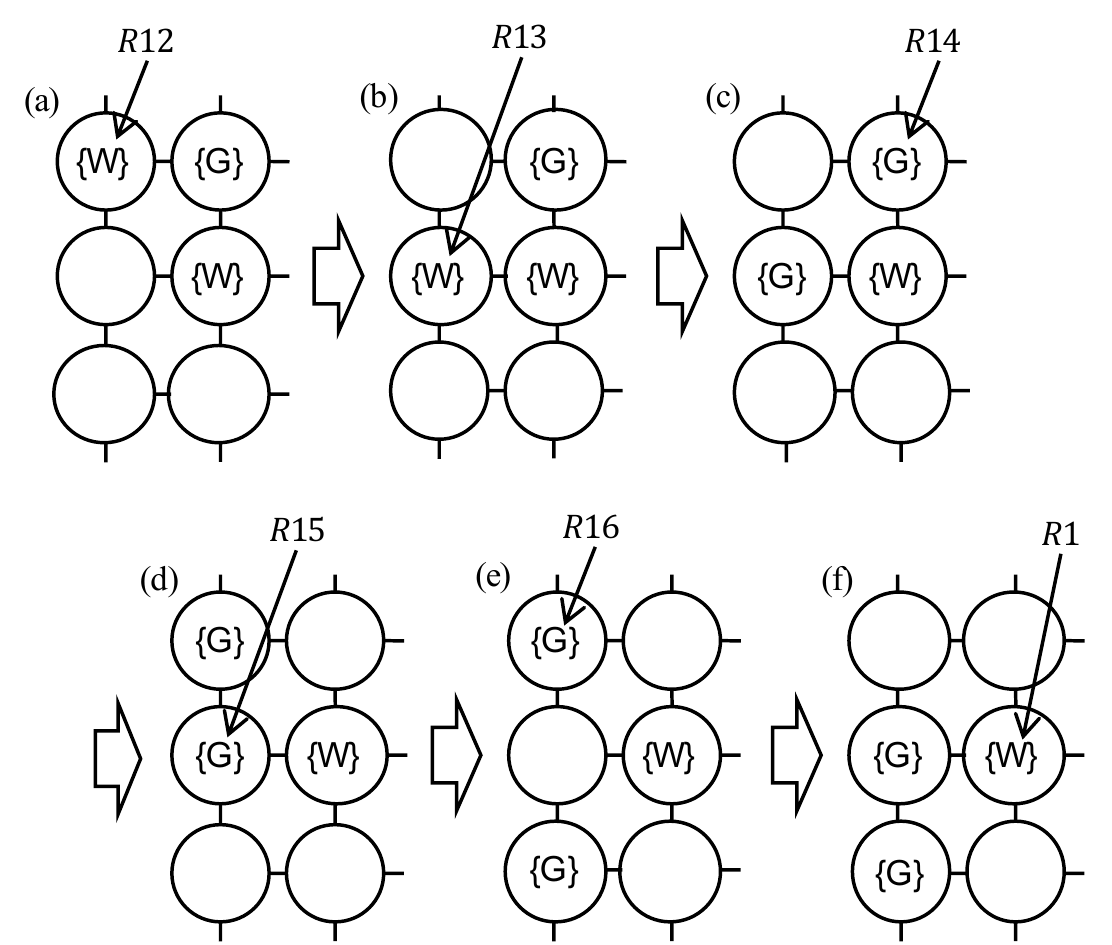}
\end{center}
\caption{Turning east in an execution of Algorithm\,\ref{algorithmA22T3}}
\label{turnEastA22T3}
\end{figure}
After robots proceed west, they reach the west end of the grid (Fig.\,\ref{turnEastA22T3}(a)). 
From this configuration, the robot with color $\W$ on a west node moves south by rule $R12$, and hence the configuration becomes one in Fig.\,\ref{turnEastA22T3}(b).
From this configuration, the robot with color $\W$ on a west node changes its color to $\G$ by rule $R13$, and hence the configuration becomes one in Fig.\,\ref{turnEastA22T3}(c).
From this configuration, the robot with color $\G$ on a north node moves west by rule $R14$, and hence the configuration becomes one in Fig.\,\ref{turnEastA22T3}(d).
From this configuration, the robot with color $\G$ on a south node moves south by rule $R15$, and hence the configuration becomes one in Fig.\,\ref{turnEastA22T3}(e).
From this configuration, the robot with color $\G$ on a north node moves south by rule $R16$, and hence the configuration becomes one in Fig.\,\ref{turnEastA22T3}(f).
From this configuration, two robots can proceed east again.

\paragraph{End of exploration.}
In case that $m$ is odd, robots visit the south end nodes while proceeding west.
Eventually, the configuration becomes $\{(v_{m-2,0},\{\W\}),(v_{m-2,1},\{\G\}),(v_{m-1,1},\{\W\})\}$.
Node $v_{m-1,0}$ has not been visited yet.
From this configuration, the robot on $v_{m-2,0}$ moves to $v_{m-1,0}$ by rule $R12$, and hence the configuration becomes $\{(v_{m-2,1},\{\G\}),(v_{m-1,0},\{\W\}),(v_{m-1,1},\{\W\})\}$.
At this configuration, no robots are enabled.
In case that $m$ is even, robots visit the south end nodes while proceeding east.
Eventually, the configuration becomes $\{(v_{m-2,n-2},\{\G\}),(v_{m-2,n-1},\{\W\}),(v_{m-1,n-2},\{\G\})\}$.
Node $v_{m-1,n-1}$ has not been visited yet.
From this configuration, the robot on $v_{m-2,n-1}$ moves to $v_{m-1,n-1}$ by rule $R4$, and hence the configuration becomes $\{(v_{m-2,n-2},\{\G\}),(v_{m-1,n-2},\{\G\}),(v_{m-1,n-1},\{\W\})\}$.
At this configuration, no robots are enabled.

\subsubsection{$\phi=2$, $\ell=2$, no common chirality, and $k=4$}
\label{secA22F4}
We give a terminating exploration algorithm for $m\times n$ grids $(m\geq2, n\geq3)$ in case of $\phi=2$, $\ell=2$, no common chirality, and $k=4$.
A set of colors is $Col=\{\G, \W\}$.
The algorithm is given in Algorithm \ref{algorithmA22F4}.

\begin{algorithm}[tbp]
\caption{Asynchronous Terminating Exploration for $\phi=2,\,\ell=2,$ $k=4$ Without Common Chirality}
\label{algorithmA22F4}
\begin{algorithmic}
\renewcommand{\algorithmicrequire}{\textbf{Initial configuration}}
\REQUIRE
\STATE $\{(v_{0,0},\{\G\}),(v_{0,1},\{\W\}),(v_{0,2},\{\W\}),(v_{1,0},\{\W\})\}$
\renewcommand{\algorithmicrequire}{\textbf{Rules}}
\REQUIRE
\STATE
  \centering
  \includegraphics[width=0.95\textwidth]{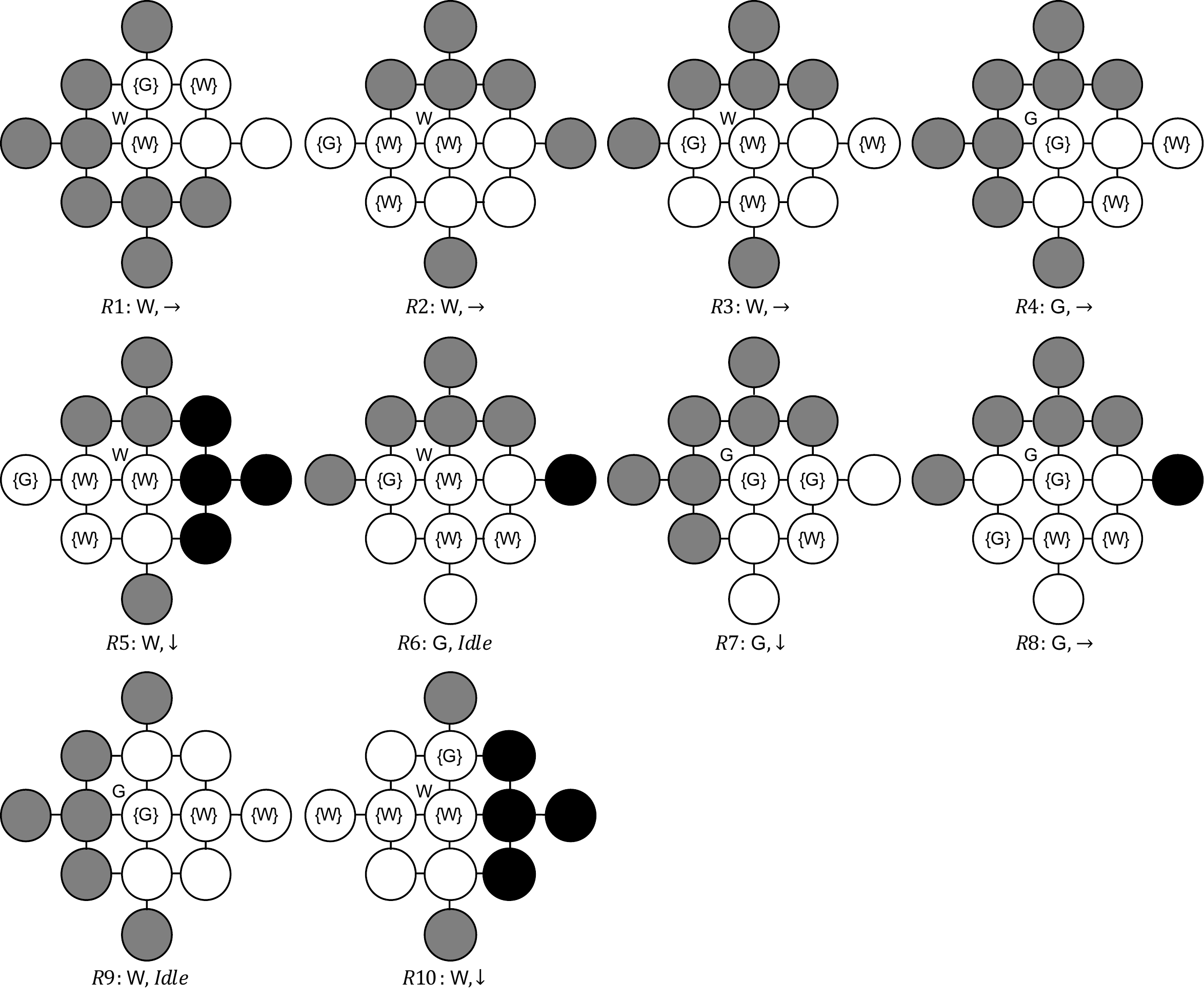}
\end{algorithmic}
\end{algorithm}

\paragraph{Proceeding east.}
The process of proceeding east is shown in Fig.\,\ref{ProceedEastA22F4}.
\begin{figure}[tbp]
\begin{center}
 \includegraphics[scale=0.8]{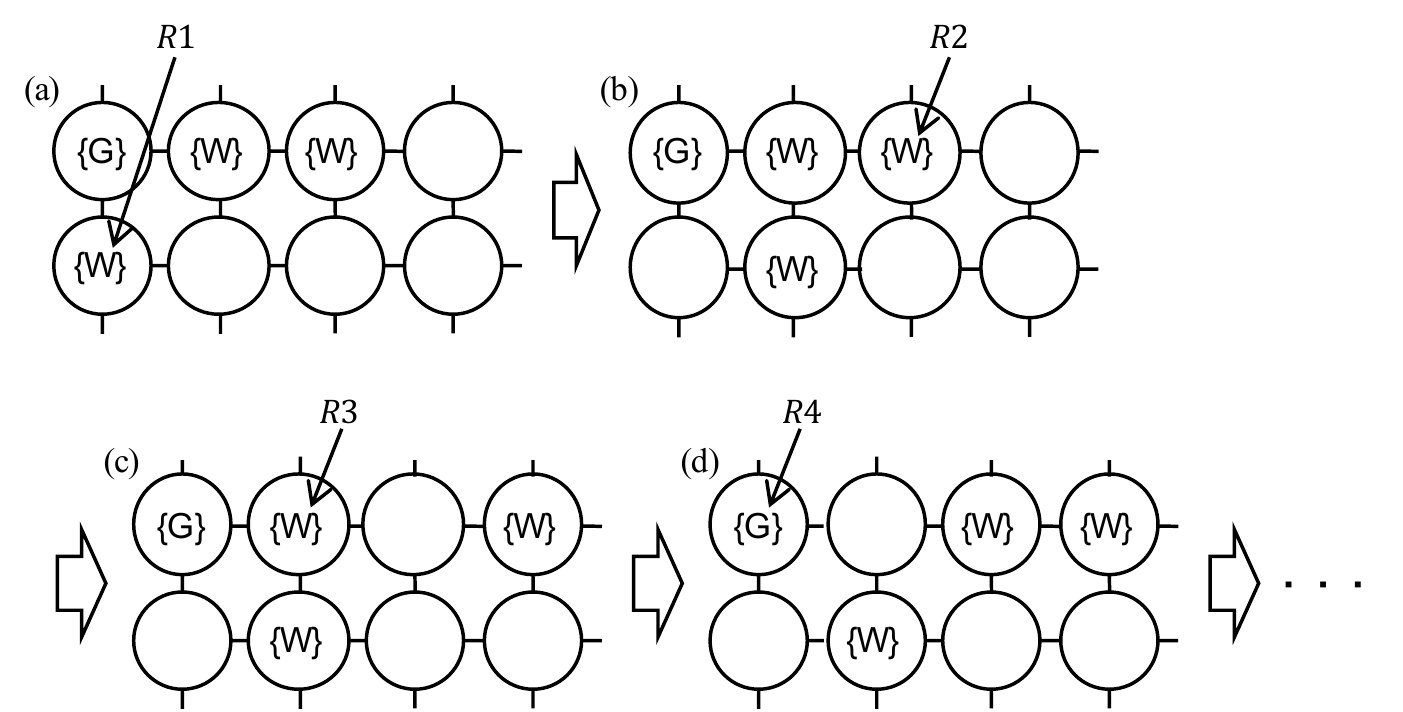}
\end{center}
\caption{Proceeding east in an execution of Algorithm\,\ref{algorithmA22F4}}
\label{ProceedEastA22F4}
\end{figure}
At the initial configuration or at a configuration immediately after turning east, robots make the form in Fig.\,\ref{ProceedEastA22F4}(a).
From this configuration, the robot with color $\W$ on a south node moves east by rule $R1$, and hence the configuration becomes one in Fig.\,\ref{ProceedEastA22F4}(b).
From this configuration, the robot with color $\W$ on an east node moves east by rule $R2$, and hence the configuration becomes one in Fig.\,\ref{ProceedEastA22F4}(c).
From this configuration, the robot with color $\W$ neighboring to the robot with color $\G$ moves east by rule $R3$, and hence the configuration becomes one in Fig.\,\ref{ProceedEastA22F4}(d).
From this configuration, the robot with color $\G$ moves east by rule $R4$.
After that, robots proceed east while keeping the form by repeatedly executing those rules. 

\paragraph{Turning west.}
The process of turning west is shown in Fig.\,\ref{turnWestA22F4}.
\begin{figure}[tbp]
\begin{center}
 \includegraphics[scale=0.8]{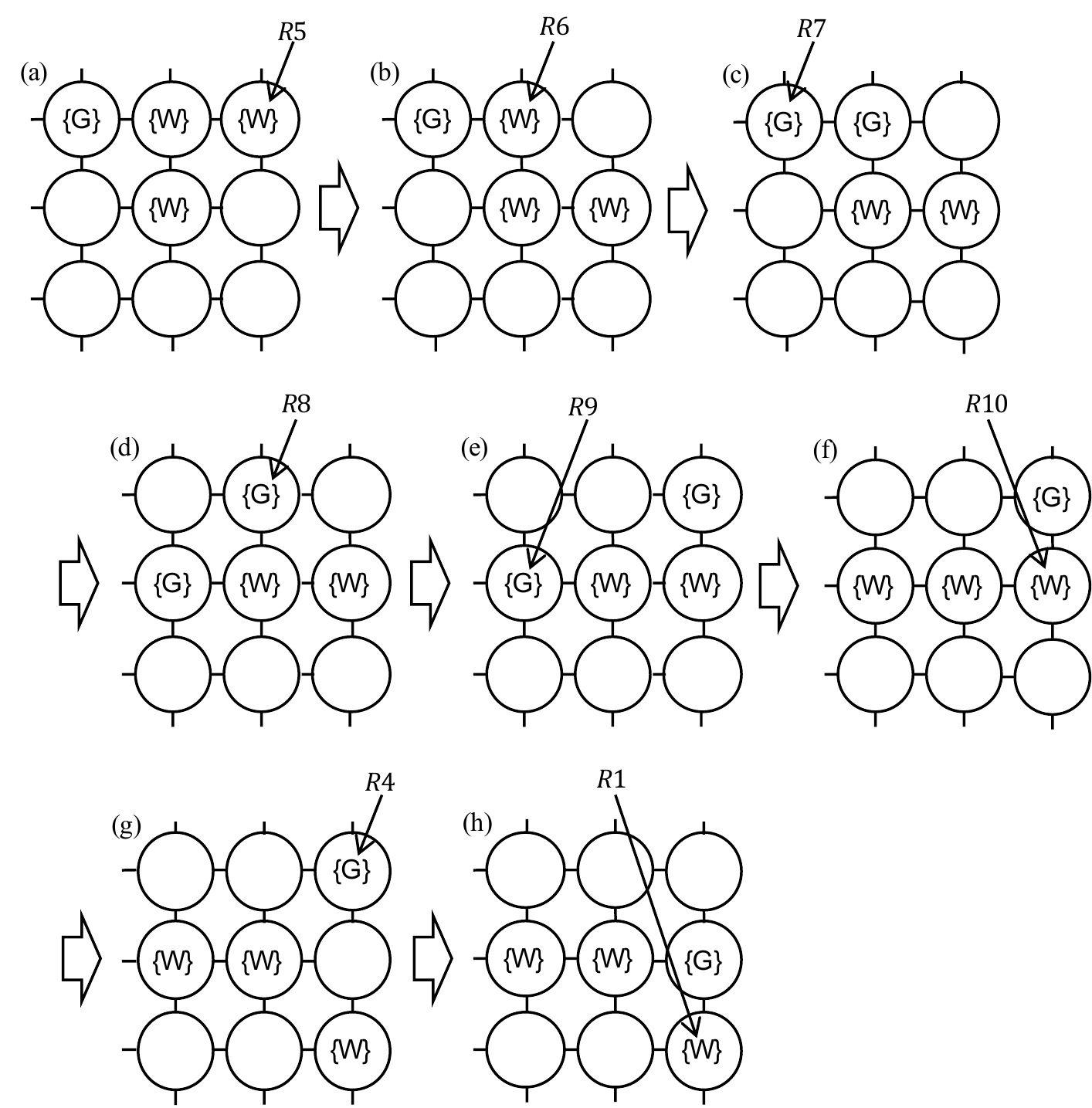}
\end{center}
\caption{Turning west in an execution of Algorithm\,\ref{algorithmA22F4}}
\label{turnWestA22F4}
\end{figure}
After robots proceed east, they reach the east end of the grid, and the configuration becomes one in Fig.\,\ref{turnWestA22F4}(a). 
From this configuration, the robot at the east end moves south by rule $R5$, and hence the configuration becomes one in Fig.\,\ref{turnWestA22F4}(b).
From this configuration, the robot with color $\W$ on a north node changes its color to $\G$ by rule $R6$, and hence the configuration becomes one in Fig.\,\ref{turnWestA22F4}(c).
From this configuration, the robot with color $\G$ on a west node moves south by rule $R7$, and hence the configuration becomes one in Fig.\,\ref{turnWestA22F4}(d).
From this configuration, the robot with color $\G$ on a north node moves east by rule $R8$, and hence the configuration becomes one in Fig.\,\ref{turnWestA22F4}(e).
From this configuration, the robot with color $\G$ on a west node changes its color to $\W$ by rule $R9$, and hence the configuration becomes one in Fig.\,\ref{turnWestA22F4}(f).
From this configuration, the robot with color $\W$ on an east node moves south by rule $R10$, and hence the configuration becomes one in Fig.\,\ref{turnWestA22F4}(g).
From this configuration, the robot with color $\G$ moves south by rule $R4$, and hence the configuration becomes one in Fig.\,\ref{turnWestA22F4}(h).

\paragraph{Proceeding west and turning east.}
The form of robots in Fig.\,\ref{turnWestA22F4}(h) is a mirror image of the one that robots make to proceed east.
Hence, robots proceed west and turn east with the same rules as proceeding east and turning west, respectively.

\paragraph{End of exploration.}
In case that $m$ is odd, robots visit the south end nodes while proceeding west.
Eventually, the configuration becomes $\{(v_{m-2,0},\{\W\}),(v_{m-2,1},\{\W\}),(v_{m-2,2},\{\G\}),(v_{m-1,1},\{\W\})\}$.
Node $v_{m-1,0}$ has not been visited yet.
From this configuration, the robot on $v_{m-2,0}$ moves to $v_{m-1,0}$ by rule $R5$, and hence the configuration becomes $\{(v_{m-2,1},\{\W\}),(v_{m-2,2},\{\G\}),(v_{m-1,0},\{\W\}),(v_{m-1,1},\{\W\})\}$.
At this configuration, no robots are enabled.
In case that $m$ is even, robots terminate the algorithm similarly to the odd case.

\subsubsection{$\phi=1$, $\ell=3$, a common chirality, and $k=3$}
\label{secA13T3}
We give a terminating exploration algorithm for $m\times n$ grids $(m\geq2, n\geq3)$ in case of $\phi=1$, $\ell=3$, a common chirality, and $k=3$.
A set of colors is $Col=\{\G, \W, \B\}$.
The algorithm is given in Algorithm \ref{algorithmA13T3}.

\begin{algorithm}[tbp]
\caption{Asynchronous Terminating Exploration for $\phi=1,\,\ell=3,$ $k=3$ with Common Chirality}
\label{algorithmA13T3}
\begin{algorithmic}
\renewcommand{\algorithmicrequire}{\textbf{Initial configuration}}
\REQUIRE
\STATE $\{(v_{0,0},\{\G\}),(v_{0,1},\{\W\}),(v_{0,2},\{\W\})\}$
\renewcommand{\algorithmicrequire}{\textbf{Rules}}
\REQUIRE
\STATE
  \centering
  \includegraphics[width=0.95\textwidth]{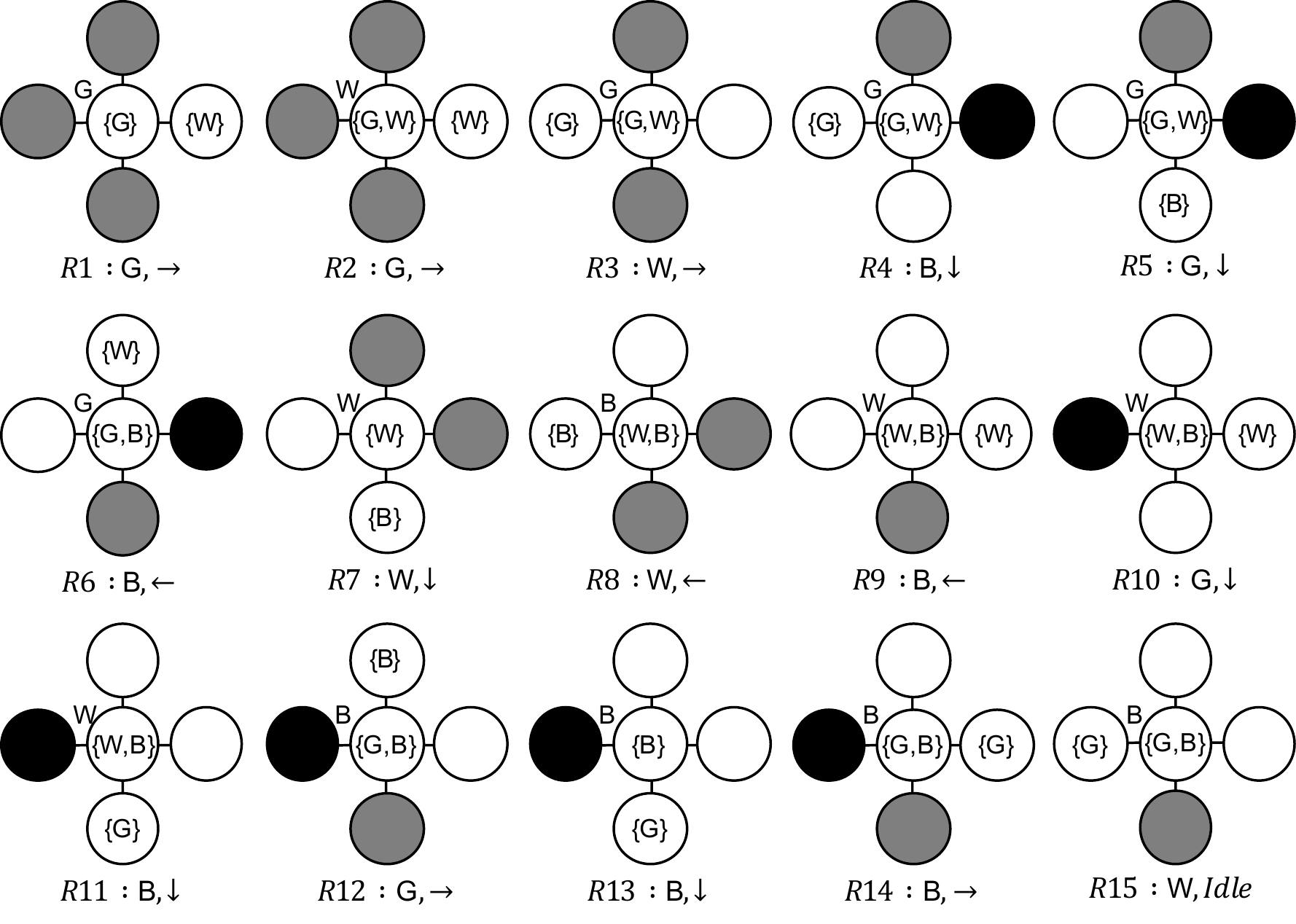}
\end{algorithmic}
\end{algorithm}

\paragraph{Proceeding east.}
The process of proceeding east is shown in Fig.\,\ref{ProceedEastA13T3}.
We use the same procedure as a ring exploration algorithm in \cite{Ooshita21:Ring}.
\begin{figure}[tbp]
\begin{center}
 \includegraphics[scale=0.8]{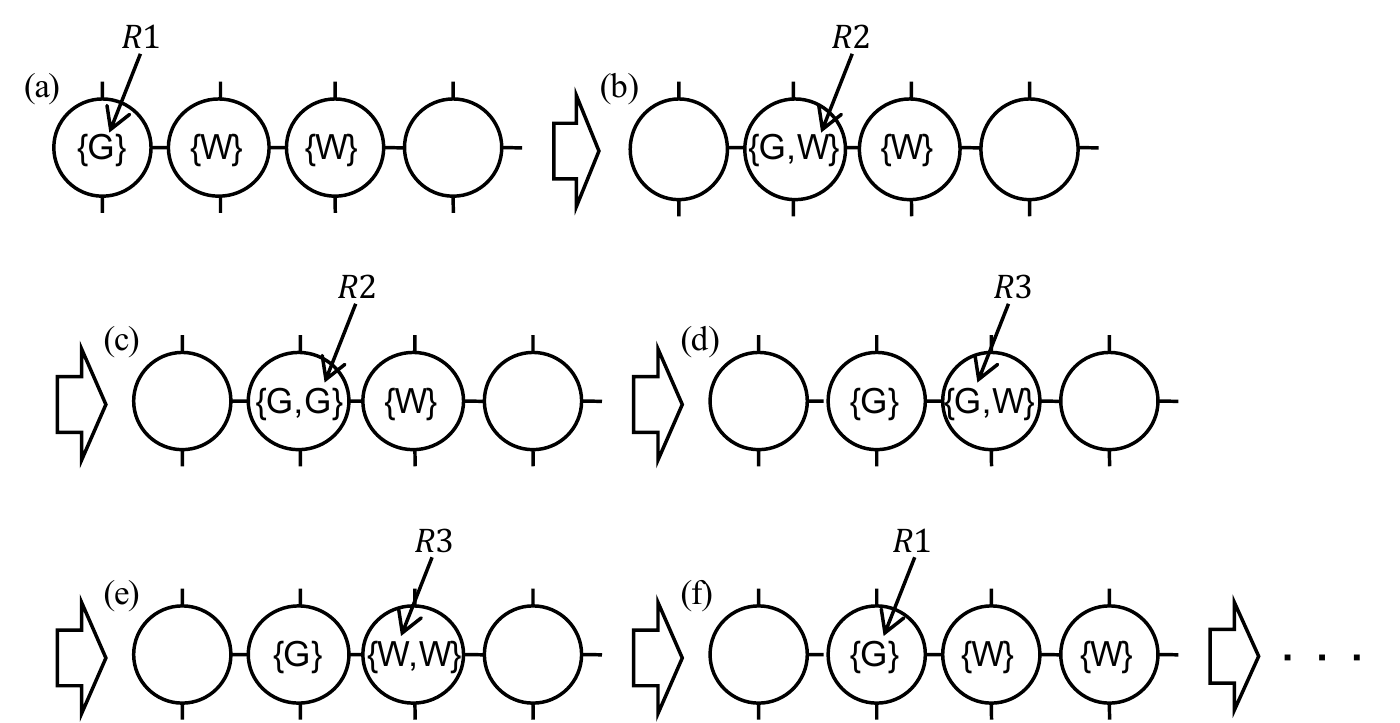}
\end{center}
\caption{Proceeding east in an execution of Algorithm\,\ref{algorithmA13T3}}
\label{ProceedEastA13T3}
\end{figure}
At the initial configuration or at a configuration immediately after turning east, robots make the form in Fig.\,\ref{ProceedEastA13T3}(a).
From this configuration, the robot with color $\G$ moves east by rule $R1$, and hence the configuration becomes one in Fig.\,\ref{ProceedEastA13T3}(b).
From this configuration, the robot with color $\W$ on a west node changes its color to $\G$ and moves east by rule $R2$.
In the ASYNC model, after it changes its color to $\G$, other robots may observe the intermediate configuration (Fig.\,\ref{ProceedEastA13T3}(c)).
However, there are no rules that the other robots can execute in the intermediate configuration.
Hence, the configuration becomes one in Fig.\,\ref{ProceedEastA13T3}(d).
From this configuration, the robot with color $\G$ on an east node changes its color to $\W$ and moves east by rule $R3$.
In the ASYNC model, after it changes its color to $\W$, other robots may observe the intermediate configuration (Fig.\,\ref{ProceedEastA13T3}(e)).
However, there are no rules that the other robots can execute in the intermediate configuration.
Hence, the configuration becomes one in Fig.\,\ref{ProceedEastA13T3}(f).
After that, robots proceed east while keeping the form by repeatedly executing those rules. 

\paragraph{Turning west.}
The process of turning west is shown in Fig.\,\ref{turnWestA13T3}.
\begin{figure}[tbp]
\begin{center}
 \includegraphics[scale=0.8]{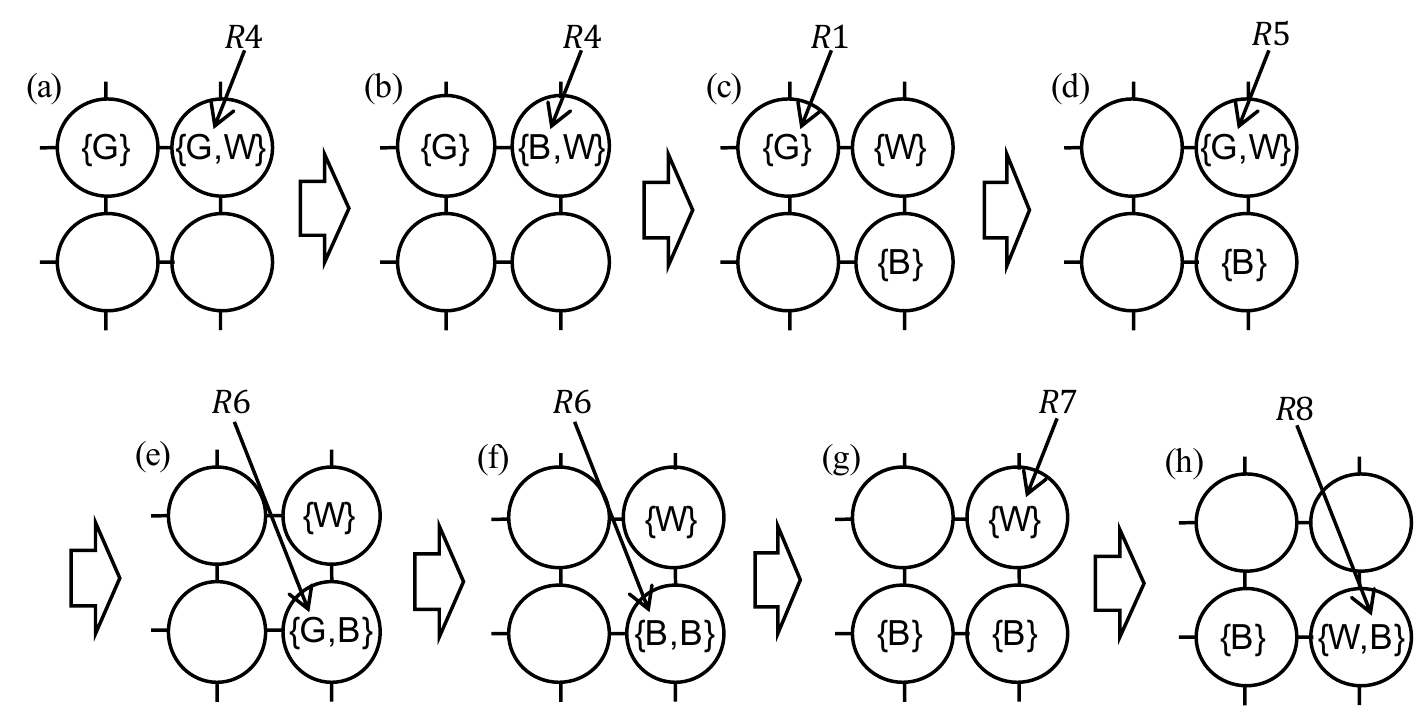}
\end{center}
\caption{Turning west in an execution of Algorithm\,\ref{algorithmA13T3}}
\label{turnWestA13T3}
\end{figure}
After robots proceed east, they reach the east end of the grid, and the configuration becomes one in Fig.\,\ref{turnWestA13T3}(a). 
From this configuration, the robot with color $\G$ on an east node changes its color to $\B$ and moves south by rule $R4$.
In the ASYNC model, after it changes its color to $\B$, other robots may observe the intermediate configuration (Fig.\,\ref{turnWestA13T3}(b)).
However, there are no rules that the other robots can execute in the intermediate configuration.
Hence, the configuration becomes one in Fig.\,\ref{turnWestA13T3}(c).
From this configuration, the robot with color $\G$ moves east by rule $R1$, and hence the configuration becomes one in Fig.\,\ref{turnWestA13T3}(d).
From this configuration, the robot with color $\G$ moves south by rule $R5$, and hence the configuration becomes one in Fig.\,\ref{turnWestA13T3}(e).
From this configuration, the robot with color $\G$ changes its color to $\B$ and moves west by rule $R6$.
In the ASYNC model, after it changes its color to $\B$, other robots may observe the intermediate configuration (Fig.\,\ref{turnWestA13T3}(f)).
However, there are no rules that the other robots can execute in the intermediate configuration.
Hence, the configuration becomes one in Fig.\,\ref{turnWestA13T3}(g).
From this configuration, the robot with color $\W$ moves south by rule $R7$, and hence the configuration becomes one in Fig.\,\ref{turnWestA13T3}(h).

\paragraph{Proceeding west.}
The process of proceeding west is similar to that of proceeding east. 
Robots with colors $\W$ and $\B$ for proceeding west move in the same way as robots with colors $\G$ and $\W$ for proceeding east, respectively. 
The form in Fig.\,\ref{turnWestA13T3}(h) corresponds to one in Fig.\,\ref{ProceedEastA13T3}(b).
Rules $R7$, $R8$, and $R9$ for proceeding west correspond to rules $R1$, $R2$, and $R3$ for proceeding east, respectively.
Hence, robots proceed west keeping the form by repeatedly executing those rules.

\paragraph{Turning east.}
The process of turning east is shown in Fig.\,\ref{turnEastA13T3}.
\begin{figure}[tbp]
\begin{center}
 \includegraphics[scale=0.8]{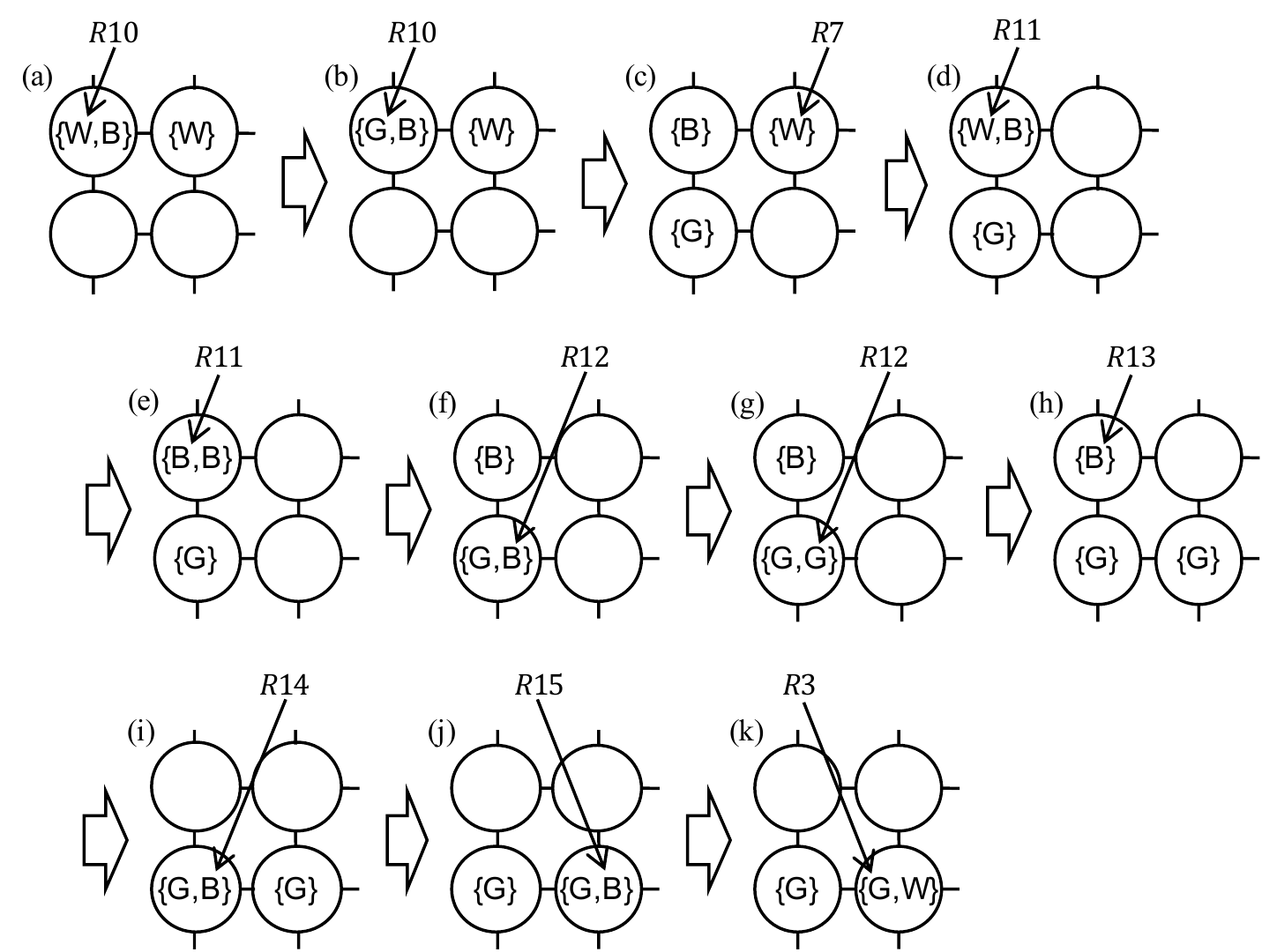}
\end{center}
\caption{Turning east in an execution of Algorithm\,\ref{algorithmA13T3}}
\label{turnEastA13T3}
\end{figure}
After robots proceed west, they reach the west end of the grid (Fig.\,\ref{turnEastA13T3}(a)). 
From this configuration, the robot with color $\W$ on a west node changes its color to $\G$ and moves south by rule $R10$.
In the ASYNC model, after it changes its color to $\W$, other robots may observe the intermediate configuration (Fig.\,\ref{turnEastA13T3}(b)).
However, there are no rules that the other robots can execute in the intermediate configuration.
Hence, the configuration becomes one in Fig.\,\ref{turnEastA13T3}(c).
From this configuration, the robot with color $\W$ moves west by rule $R7$, and hence the configuration becomes one in Fig.\,\ref{turnEastA13T3}(d).
From this configuration, the robot with color $\W$ changes its color to $\B$ and moves south by rule $R11$.
In the ASYNC model, after it changes its color to $\B$, other robots may observe the intermediate configuration (Fig.\,\ref{turnEastA13T3}(e)).
However, there are no rules that the other robots can execute in the intermediate configuration.
Hence, the configuration becomes one in Fig.\,\ref{turnEastA13T3}(f).
From this configuration, the robot with color $\B$ on a south node changes its color to $\G$ and moves east by rule $R12$.
In the ASYNC model, after it changes its color to $\G$, other robots may observe the intermediate configuration (Fig.\,\ref{turnEastA13T3}(g)).
However, there are no rules that the other robots can execute in the intermediate configuration.
Hence, the configuration becomes one in Fig.\,\ref{turnEastA13T3}(h).
From this configuration, the robot with color $\B$ moves south by rule $R13$, and hence the configuration becomes one in Fig.\,\ref{turnEastA13T3}(i).
From this configuration, the robot with color $\B$ moves east by rule $R14$, and hence the configuration becomes one in Fig.\,\ref{turnEastA13T3}(j).
From this configuration, the robot with color $\B$ changes its color to $\W$ by rule $R15$, and hence the configuration becomes one in Fig.\,\ref{turnEastA13T3}(k).
From this configuration, robots can proceed east again since their form is the same as one in Fig.\,\ref{ProceedEastA13T3}(d).

\paragraph{End of exploration.}
In case that $m$ is odd, robots visit the south end nodes while proceeding east.
Eventually, the configuration becomes $\{(v_{m-1,n-2},\{\G\}),(v_{m-1,n-1},\{\G,\W\})\}$.
At this configuration, no robots are enabled.
In case that $m$ is even, robots visit the south end nodes while proceeding east.
Eventually, the configuration becomes $\{(v_{m-1,0},\{\W,\B\}),(v_{m-1,1},\{\W\})\}$.
At this configuration, no robots are enabled.

\subsubsection{$\phi=1$, $\ell=3$, no common chirality, and $k=6$}
\label{secA13F6}
We give a terminating exploration algorithm for $m\times n$ grids $(m\geq3, n\geq3)$ in case of $\phi=1$, $\ell=3$, no common chirality, and $k=6$.
A set of colors is $Col=\{\G, \W, \B\}$.
The algorithm is given in Algorithm \ref{algorithmA13F6}.

\begin{algorithm}[tbp]
\caption{Asynchronous Terminating Exploration for $\phi=1,\,\ell=3,$ $k=6$ Without Common Chirality}
\label{algorithmA13F6}
\begin{algorithmic}
\renewcommand{\algorithmicrequire}{\textbf{Initial configuration}}
\REQUIRE
\STATE $\{(v_{0,0},\{\G\}),(v_{0,1},\{\W\}),(v_{0,2},\{\W\}),(v_{1,0},\{\W,\B\}),(v_{1,1},\{\W\})\}$
\renewcommand{\algorithmicrequire}{\textbf{Rules}}
\REQUIRE
\STATE
  \centering
  \includegraphics[width=0.95\textwidth]{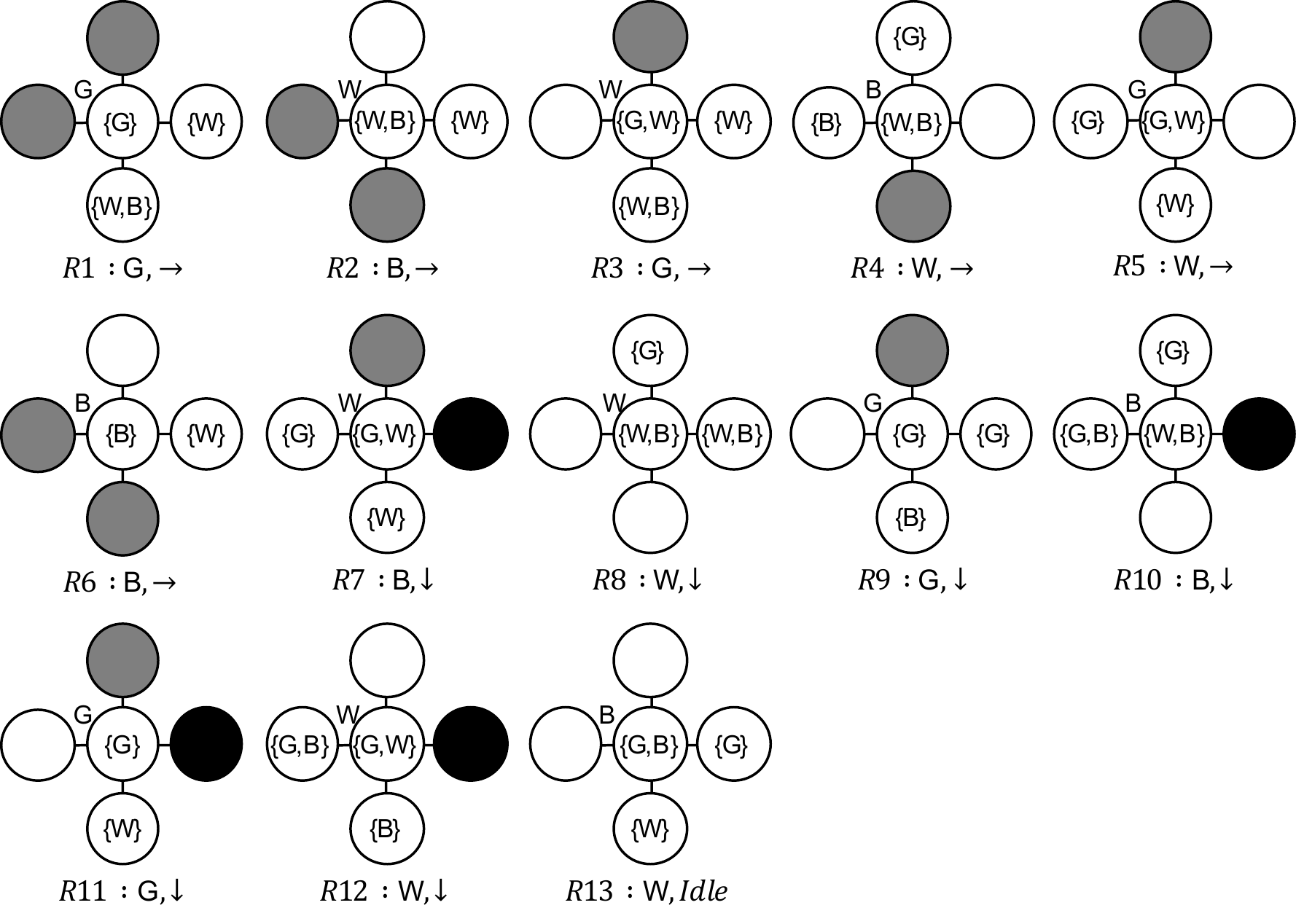}
\end{algorithmic}
\end{algorithm}

\paragraph{Proceeding east.}
The process of proceeding east is shown in Fig.\,\ref{ProceedEastA13F6-1} and Fig.\,\ref{ProceedEastA13F6-2}.
\begin{figure}[tbp]
\begin{center}
 \includegraphics[scale=0.8]{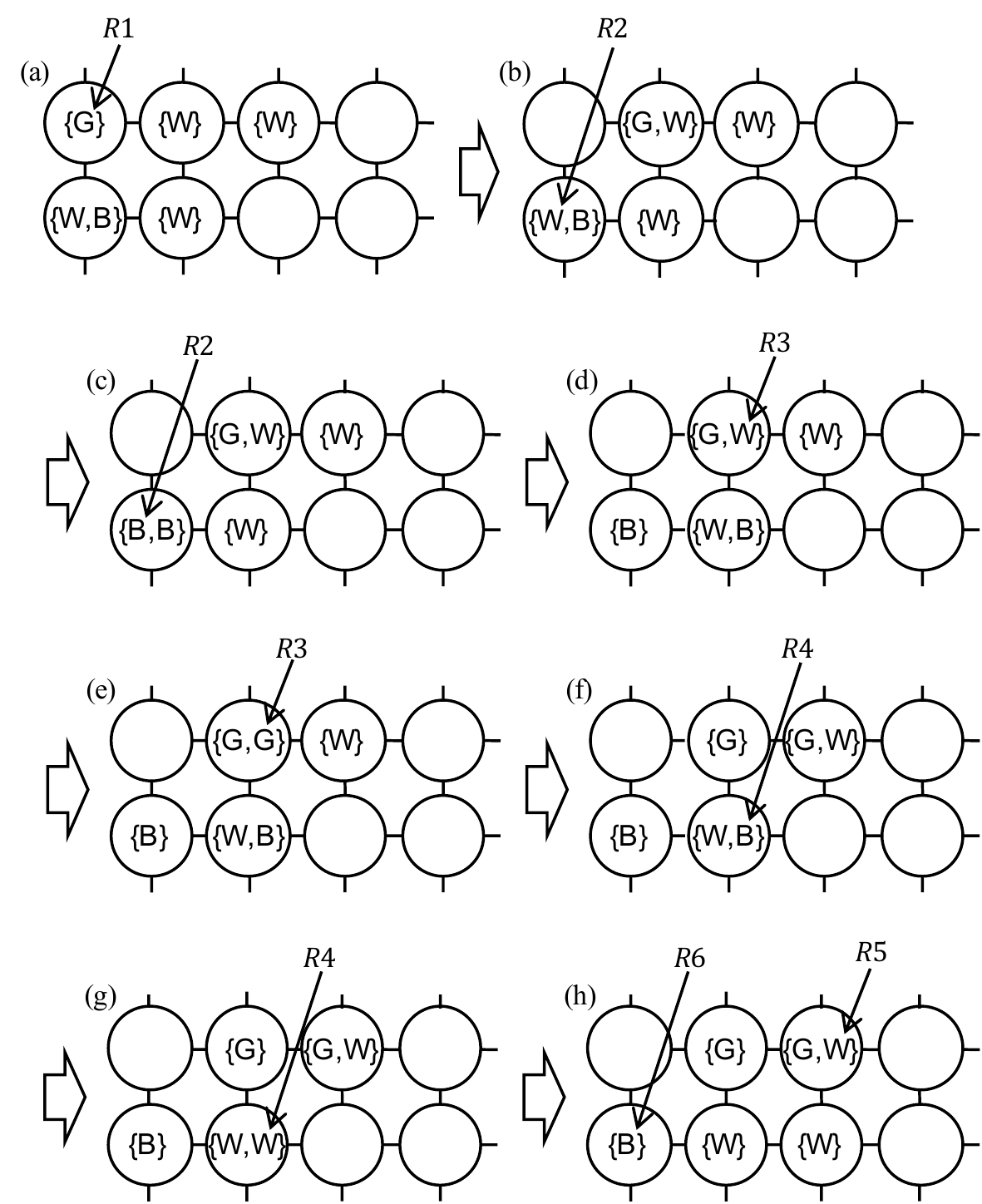}
\end{center}
\caption{Proceeding east in executions of Algorithm\,\ref{algorithmA13F6} (I)}
\label{ProceedEastA13F6-1}
\end{figure}
\begin{figure}[tbp]
\begin{center}
 \includegraphics[scale=0.8]{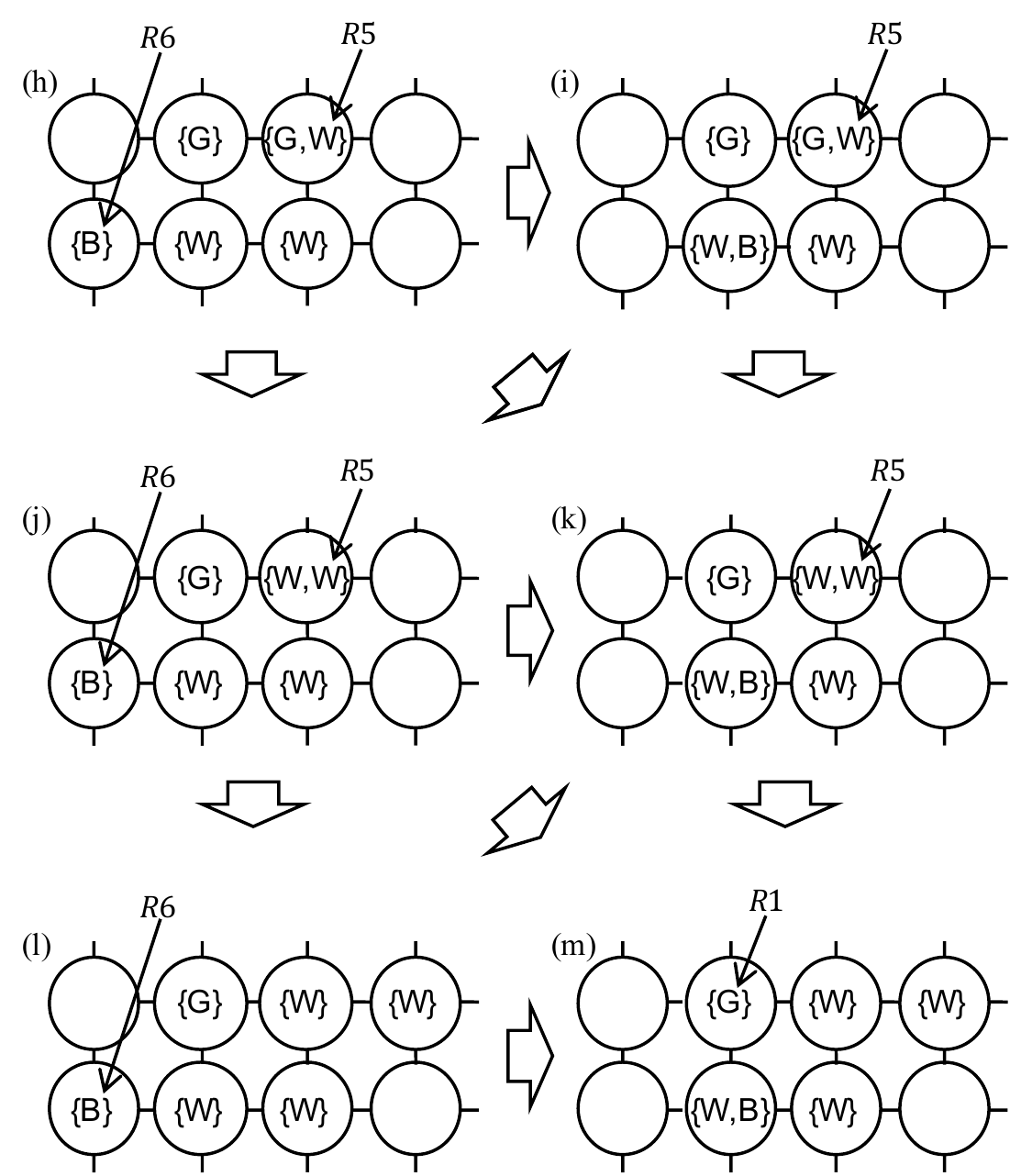}
\end{center}
\caption{Proceeding east in executions of Algorithm\,\ref{algorithmA13F6} (I\hspace{-1pt}I)}
\label{ProceedEastA13F6-2}
\end{figure}
At the initial configuration or at a configuration immediately after turning east, robots make the form in Fig.\,\ref{ProceedEastA13F6-1}(a).
From this configuration, the robot with color $\G$ moves east by rule $R1$, and hence the configuration becomes one in Fig.\,\ref{ProceedEastA13F6-1}(b).
From this configuration, the robot with color $\W$ on a west node changes its color to $\B$ and moves east by rule $R2$.
In the ASYNC model, after it changes its color to $\B$, other robots may observe the intermediate configuration (Fig.\,\ref{ProceedEastA13F6-1}(c)).
However, there are no rules that the other robots can execute in the intermediate configuration.
Hence, the configuration becomes one in Fig.\,\ref{ProceedEastA13F6-1}(d).
From this configuration, the robot with color $\W$ occupying the same node as the robot with color $\G$ changes its color to $\G$ and moves east by rule $R3$.
In the ASYNC model, after it changes its color to $\G$, other robots may observe the intermediate configuration (Fig.\,\ref{ProceedEastA13F6-1}(e)).
However, there are no rules that the other robots can execute in the intermediate configuration.
Hence, the configuration becomes one in Fig.\,\ref{ProceedEastA13F6-1}(f).
From this configuration, the robot with color $\B$ occupying the same node as the robot with color $\W$ changes its color to $\W$ and moves east by rule $R4$.
In the ASYNC model, after it changes its color to $\W$, other robots may observe the intermediate configuration (Fig.\,\ref{ProceedEastA13F6-1}(g)).
However, there are no rules that the other robots can execute in the intermediate configuration.
Hence, the configuration becomes one in Fig.\,\ref{ProceedEastA13F6-1}(h).

Fig.\,\ref{ProceedEastA13F6-2}(h) denotes the same configuration as one in Fig.\,\ref{ProceedEastA13F6-1}(h).
We show that the configuration eventually becomes one in Fig.\,\ref{ProceedEastA13F6-2}(m) regardless of the scheduler.
At the configuration in Fig.\,\ref{ProceedEastA13F6-2}(h), let $r_1$ be the robot with color $\W$ on a northeast node and let $r_2$ be the robot with color $\B$.
Then, $r_1$ can execute rule $R5$, and $r_2$ can execute rule $R6$.
If $r_2$ finishes $R6$ before $r_1$ finishes the compute phase of $R5$, the configuration becomes one in Fig.\,\ref{ProceedEastA13F6-2}(i).
If $r_1$ finishes the compute phase of $R5$ before $r_2$ finishes $R6$, the configuration becomes one in Fig.\,\ref{ProceedEastA13F6-2}(j).
If $r_1$ finishes the compute phase of $R5$ and $r_2$ finishes $R6$ at the same time, the configuration becomes one in Fig.\,\ref{ProceedEastA13F6-2}(k).
At the configurations in Fig.\,\ref{ProceedEastA13F6-2}(i) and Fig.\,\ref{ProceedEastA13F6-2}(k), robots cannot execute rules except $R5$, and hence the configuration eventually becomes one in Fig.\,\ref{ProceedEastA13F6-2}(m).
At the configuration in Fig.\,\ref{ProceedEastA13F6-2}(j), robots cannot execute rules except $R5$ and $R6$.
From this configuration, if $r_2$ finishes $R6$ before $r_1$ finishes $R5$, the configuration becomes one in Fig.\,\ref{ProceedEastA13F6-2}(k).
If $r_1$ finishes $R5$ before $r_2$ finishes $R6$, the configuration becomes one in Fig.\,\ref{ProceedEastA13F6-2}(l).
If $r_1$ finishes $R5$ and $r_2$ finishes $R6$ at the same time, the configuration becomes one in Fig.\,\ref{ProceedEastA13F6-2}(m).
At the configurations in Fig.\,\ref{ProceedEastA13F6-2}(l), robots cannot execute rules except $R6$, and hence the configuration eventually becomes one in Fig.\,\ref{ProceedEastA13F6-2}(m).
From the above discussion, the configuration eventually becomes one in Fig.\,\ref{ProceedEastA13F6-2}(m) in any case.
In this configuration, the form of robots is the same as in Fig.\,\ref{ProceedEastA13F6-1}(a).
Hence, robots proceed east while keeping their form by repeatedly executing those rules. 

\paragraph{Turning west.}
The process of turning west is shown in Fig.\,\ref{turnWestA13F6-1} and Fig.\,\ref{turnWestA13F6-2}.
\begin{figure}[tbp]
\begin{center}
 \includegraphics[scale=0.8]{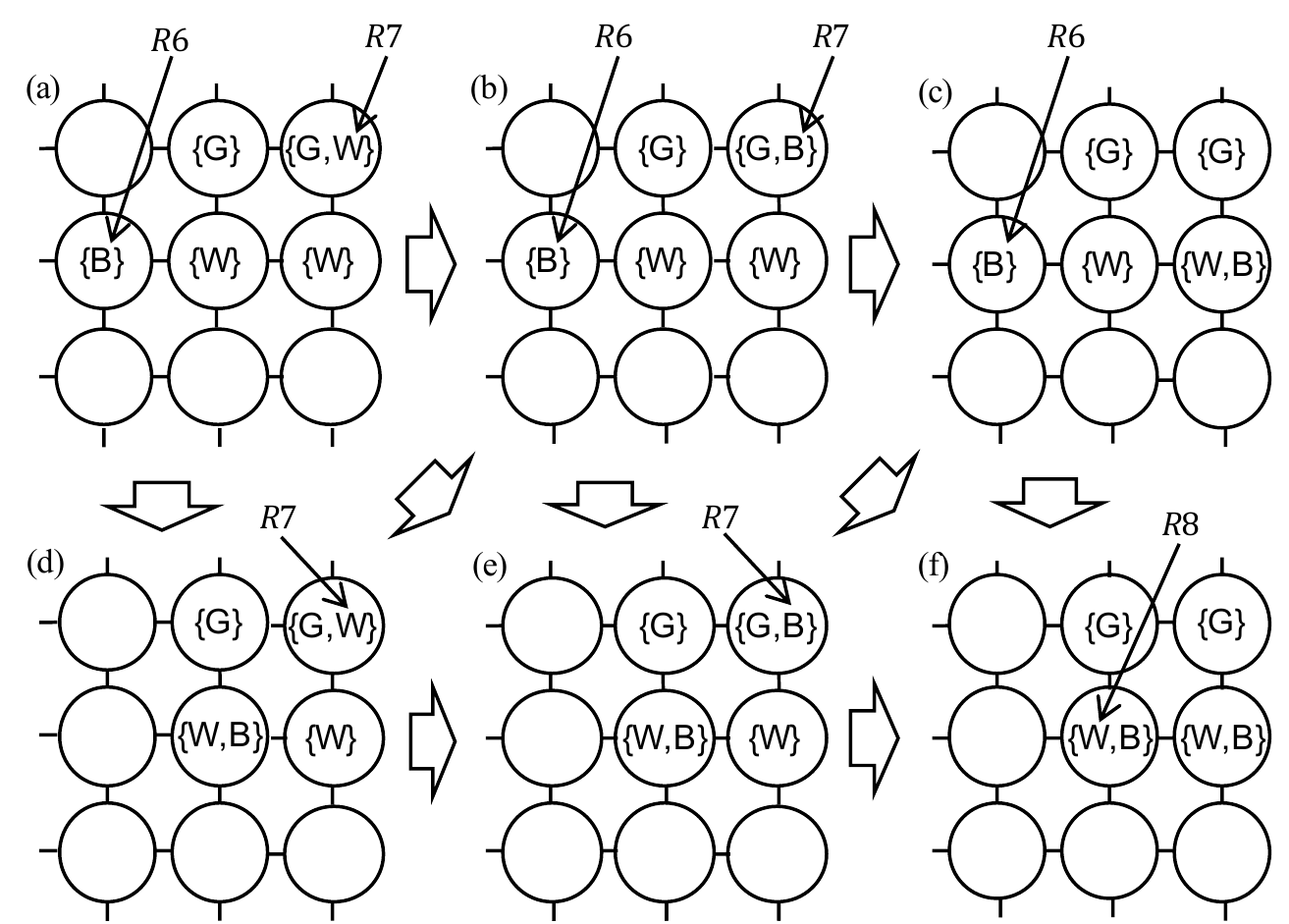}
\end{center}
\caption{Turning west in an execution of Algorithm\,\ref{algorithmA13F6} (I)}
\label{turnWestA13F6-1}
\end{figure}
\begin{figure}[tbp]
\begin{center}
 \includegraphics[scale=0.8]{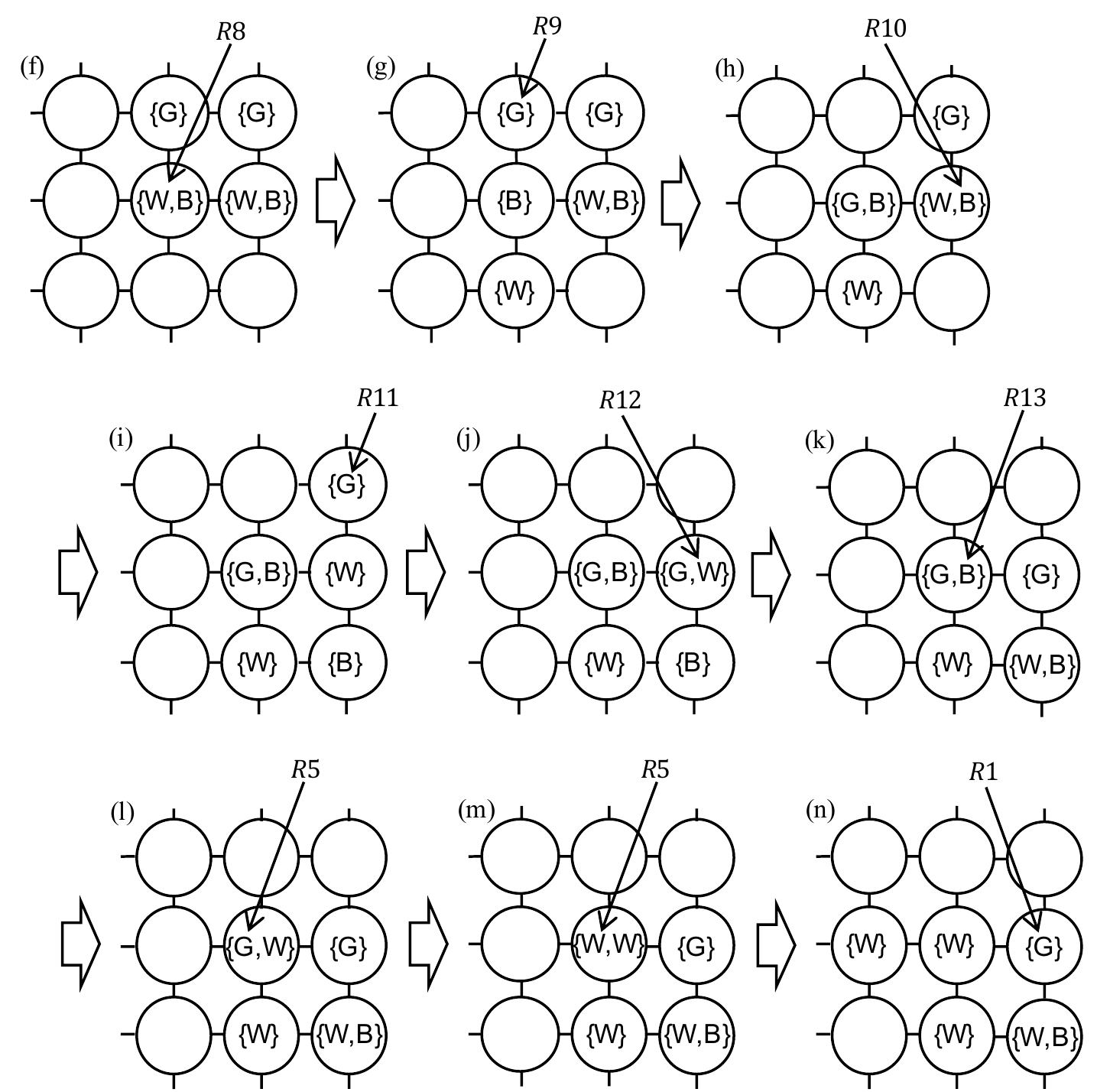}
\end{center}
\caption{Turning west in an execution of Algorithm\,\ref{algorithmA13F6} (I\hspace{-1pt}I)}
\label{turnWestA13F6-2}
\end{figure}
After robots proceed east, they reach the east end of the grid, and the configuration becomes one in Fig.\,\ref{turnWestA13F6-1}(a). 
At this configuration, let $r_1$ be the robot with color $\B$, and let $r_2$ be the robot with color $\G$ on a northeast node.
Then, $r_1$ can execute rule $R6$, and $r_2$ can execute rule $R7$.
If $r_2$ finishes the compute phase of $R7$ before $r_1$ finishes $R6$, the configuration becomes one in Fig.\,\ref{turnWestA13F6-1}(b).
If $r_1$ finishes $R6$ before $r_2$ finishes the compute phase of $R7$, the configuration becomes one in Fig.\,\ref{turnWestA13F6-1}(d).
If $r_1$ finishes $R6$ and $r_2$ finishes the compute phase of $R7$ at the same time, the configuration becomes one in Fig.\,\ref{turnWestA13F6-1}(e).
At the configurations in Fig.\,\ref{turnWestA13F6-1}(d) and Fig.\,\ref{turnWestA13F6-1}(e), robots cannot execute rules except $R7$, and hence the configuration eventually becomes one in Fig.\,\ref{turnWestA13F6-1}(f).
At the configuration in Fig.\,\ref{turnWestA13F6-1}(b), robots cannot execute rules except $R6$ and $R7$.
From this configuration, if $r_2$ finishes $R7$ before $r_1$ finishes $R6$, the configuration becomes one in Fig.\,\ref{turnWestA13F6-1}(c).
If $r_1$ finishes $R6$ before $r_2$ finishes $R7$, the configuration becomes one in Fig.\,\ref{turnWestA13F6-1}(e).
If $r_1$ finishes $R6$ and $r_2$ finishes $R7$ at the same time, the configuration becomes one in Fig.\,\ref{turnWestA13F6-1}(f).
At the configuration in Fig.\,\ref{turnWestA13F6-1}(c), robots cannot execute rules except $R6$, and hence the configuration eventually becomes one in Fig.\,\ref{turnWestA13F6-1}(f).

Fig.\,\ref{turnWestA13F6-2}(f) denotes the same configuration as one in Fig.\,\ref{turnWestA13F6-1}(f).
From this configuration, the robot with color $\W$ on a southwest node moves south by rule $R8$, and hence the configuration becomes one in Fig.\,\ref{turnWestA13F6-2}(g).
From this configuration, the robot with color $\G$ on a northwest node moves south by rule $R9$, and hence the configuration becomes one in Fig.\,\ref{turnWestA13F6-2}(h).
From this configuration, the robot with color $\B$ on an east node moves south by rule $R10$, and hence the configuration becomes one in Fig.\,\ref{turnWestA13F6-2}(i).
From this configuration, the robot with color $\G$ on an east node moves south by rule $R11$, and hence the configuration becomes one in Fig.\,\ref{turnWestA13F6-2}(j).
From this configuration, the robot with color $\W$ on an east node moves south by rule $R12$, and hence the configuration becomes one in Fig.\,\ref{turnWestA13F6-2}(k).
From this configuration, the robot with color $\B$ on a west node changes its color to $\W$ by rule $R13$, and hence the configuration becomes one in Fig.\,\ref{turnWestA13F6-2}(l).
From this configuration, the robot with color $\G$ on a northwest node changes its color to $\W$ and moves west by rule $R5$.
In the ASYNC model, after it changes its color to $\W$, other robots may observe the intermediate configuration (Fig.\,\ref{turnWestA13F6-2}(m)).
However, there are no rules that the other robots can execute in the intermediate configuration.
Hence, the configuration becomes one in Fig.\,\ref{turnWestA13F6-2}(n).

\paragraph{Proceeding west and turning east.}
The form of robots in Fig.\,\ref{turnWestA13F6-2}(n) is a mirror image of the one that robots make to proceed east.
Hence, robots proceed west and turn east with the same rules as proceeding east and turning west, respectively.

\paragraph{End of exploration.}
In case that $m$ is odd, robots visit the south end nodes while proceeding west.
Eventually, the configuration becomes $\{(v_{m-2,0},\{\G\}),(v_{m-2,1},\{\G\}),(v_{m-1,0},\{\W,\B\}),(v_{m-1,1},\{\W,\B\})\}$.
At this configuration, no robots are enabled.
In case that $m$ is even, robots terminate the algorithm similarly to the odd case.

\section{Conclusions}
In this paper, we have investigated terminating exploration algorithms for myopic robots in finite grids.
First, we have proved that, in the SSYNC and ASYNC models, three myopic robots are necessary to achieve the terminating exploration of a grid if $\phi=1$ holds.
Second, we have proposed fourteen algorithms to achieve the terminating exploration of a grid in various assumptions of synchrony, visible distance, the number of colors, and a chirality.
To the best of our knowledge, they are the first algorithms that achieve the terminating exploration of a grid by myopic robots with at most three colors and/or with no common chirality.
In addition, six proposed algorithms are optimal in terms of the number of robots.

For the future work, it is interesting to close the gap between the lower and upper bounds of the number of required robots.
It is also interesting to consider other tasks and topologies with myopic luminous robots.

\bibliographystyle{plain}
\bibliography{gridpaper-c}

\end{document}